\documentclass[reprint,floatfix,aps,longbibliography]{revtex4-2}
\usepackage{physics}
\usepackage{amsmath,amssymb}
\usepackage[colorlinks,linkcolor=blue,anchorcolor=red,citecolor=red]{hyperref}
\usepackage[capitalize]{cleveref}
\usepackage{xcolor}
\usepackage{graphicx}
\usepackage{mathtools}
\usepackage{subcaption}
\usepackage{amsthm}

\newcommand{\cart}{\mathbin{\square}}

\newcommand{\ceil}[1]{\left\lceil #1 \right\rceil}

\newtheorem{theorem}{Theorem}
\newtheorem{lemma}{Lemma}
\newcommand{\bigo}[1]{\mathcal{O}(#1)}
\newcommand{\uw}[1]{U_w\mathopen{}(#1)\mathclose{}}
\newcommand{\uf}[1]{U_f\mathopen{}(#1)\mathclose{}}
\DeclareMathOperator{\diag}{diag}

\DeclarePairedDelimiter{\nint}\lfloor\rceil
\DeclareMathOperator{\spn}{span}
\newcommand{\cg}{$\mathcal{CG}$ }

\begin{document}

\title{A framework for optimal quantum spatial search using alternating phase-walks}

\author{S. Marsh}
\email{samuel.marsh@research.uwa.edu.au}
\affiliation{Department of Physics, The University of Western Australia, Perth, Australia}

\author{J. B. Wang}
\email{jingbo.wang@uwa.edu.au}
\affiliation{Department of Physics, The University of Western Australia, Perth, Australia}

\date{\today}

\begin{abstract}
	We present a novel methodological framework for quantum spatial search, generalising the Childs \& Goldstone ($\mathcal{CG}$) algorithm via alternating applications of marked-vertex phase shifts and continuous-time quantum walks. We determine closed form expressions for the optimal walk time and phase shift parameters for periodic graphs. These parameters are chosen to rotate the system about subsets of the graph Laplacian eigenstates, amplifying the probability of measuring the marked vertex. 
	%Regardless of the dimensionality of the graphs, by careful choice of walk times the analysis can be restricted to take place exactly in a three-dimensional subspace at all stages. 
	The state evolution is asymptotically optimal for any class of periodic graphs having a fixed number of unique eigenvalues. We demonstrate the effectiveness of the algorithm by applying it to obtain $\mathcal{O}(\sqrt{N})$ search on a variety of graphs. One important class is the $n \times n^3$ rook graph, which has $N=n^4$ vertices. On this class of graphs the $\mathcal{CG}$ algorithm performs suboptimally, achieving only $\mathcal{O}(N^{-1/8})$ overlap after time $\mathcal{O}(N^{5/8})$. Using the new alternating phase-walk framework, we show that $\mathcal{O}(1)$ overlap is obtained in $\mathcal{O}(\sqrt{N})$ phase-walk iterations. 
	%The promising results motivate further study into the alternating walks formulation and its relation to the \cg algorithm.
\end{abstract}

\maketitle

\section{Introduction}

Continuous-time quantum walks (CTQWs)
%, the analogue of classical continuous-time random walks, 
are widely used as key components of quantum algorithms \cite{Childs2003,Ambainis2003,Venegas-Andraca2012}. In 2004, \citet{Childs2004} introduced the well-known \cg spatial search algorithm, using a perturbed quantum walk to locate a marked vertex on an $N$-vertex graph. For certain classes of graphs, the algorithm is able to obtain $\bigo{1}$ overlap with the marked vertex after $\bigo{\sqrt{N}}$ time, achieving asymptotically optimal quantum scaling. The hypercube, complete graph, lattices with dimension greater than 4, the star graph \cite{Cattaneo2018}, and almost all Erdos-Renyi random graphs \cite{Chakraborty2016} are some of the graphs that admit efficient quantum search under the \cg algorithm. Despite significant recent progress \cite{Chakraborty2020}, the necessary and sufficient conditions for efficient \cg spatial search have not been fully characterised, generally requiring an instance-specific analysis of the graph under consideration.

In this work, we construct a new formalism for quantum spatial search using an alternating series of continuous-time quantum walks and marked vertex phase shifts. Consider a graph with Laplacian $\mathcal{L}=D-A$, where $D$ is a diagonal matrix listing the vertex degrees and $A$ is the adjacency matrix. The general quantum state evolution starting in the equal superposition $\ket{s}$ over all vertices is defined as
\begin{equation}
    \ket{\Vec{t}, \Vec{\theta}} = \left( \prod_{k=1}^{p} \uw{t_k} \uf{\theta_k} \right) \ket{s} \, ,
    \label{eq:alternatingwalk}
\end{equation}
where $\uf\theta = e^{-i \theta \ket{\omega}\bra{\omega}}$ uses the oracular `marking' Hamiltonian to apply a relative phase shift $\theta$ to the marked vertex $\ket{\omega}$, and $\uw t = e^{-i t \mathcal{L}}$ performs a quantum walk for time $t$. To establish efficient quantum search, the aim is to determine values for $\Vec{t}$ and $\Vec{\theta}$ such that $|\langle\omega|\Vec{t}, \Vec{\theta}\rangle|^2 = \bigo 1$ with $p = \mathcal{O}(\sqrt{N})$. This formulation can be considered as a generalisation of the gate-model implementation of the \cg algorithm \cite{Roland2003}, and thus all graphs that admit optimal search under the \cg formalism also admit optimal search under the alternating phase-walk formalism.

Furthermore, the database search problem can be thought of as a simple combinatorial optimisation problem, where the objective function $f(x)$ is 1 for the marked element and 0 otherwise. The proposed alternating phase-walk framework can thus be viewed as the Quantum Walk Optimisation Algorithm (QWOA) \cite{Marsh2020} applied to search. The topic of quantum combinatorial optimisation is currently of high research interest in both the theoretical and experimental domain \cite{farhi2014quantum,a12020034,Wang2020,Marsh2019,Zhou2020,Harrigan2021,Pagano25396}, due to the applicability to near-term noisy quantum hardware. In quantum approximate optimisation, different mixers (or in the case of the QWOA, different graphs) are used to encode the constraints of the problem \cite{a12020034,Wang2020,Marsh2019}. This work provides a path towards analytically characterising the performance of different mixers on optimisation problems, as well as finding in closed form the optimal values for the variational parameters involved in these algorithms.

In this paper, we restrict to periodic graphs. A graph is periodic if there exists a time $\tau > 0$ such that for every vertex $\ket{v}$, $\abs{\bra{v}\uw{\tau} \ket{v}} = 1$. A graph exhibits periodicity if and only if the non-zero Laplacian eigenvalues are all rational multiples of each other \cite{Godsil2011}. Furthermore, in most cases we will restrict to vertex-transitive graphs so that the same state evolution applies independently of which vertex is marked; in this case the graph must have integer eigenvalues \cite{Omran2009}. An intuitive property of periodic graphs is that $\uw{t+2\pi}=\uw{t}$, and thus the total walk time in \cref{eq:alternatingwalk} is bounded by $2 \pi p$.

Our aim is to design an alternating phase-walk state evolution for periodic graphs, mapping the initial equal superposition to the marked vertex using quantum walks and marked vertex phase shifts. We derive the optimal walk times $\Vec{t}$ and phase shifts $\Vec{\theta}$ in exact closed form based on spectral analysis of the (unperturbed) Laplacian. The worst-case complexity of the state evolution is characterised by an integer $d$, which is bounded above by the number of unique graph eigenvalues (with details given in \cref{sec:spatialsearchgeneral}). In the worst-case, the total number of iterations is approximately $\frac{1}{2}\left(\frac{\pi}{2}\right)^d \sqrt{N}$ for large graphs, making the scaling asymptotically optimal when $d$ is independent of $N$. The state evolution can be expressed as a series of Householder reflections, in much the same way as the well-known Grover database search algorithm consists of reflections about the marked vertex and the equal superposition \cite{Grover1996}. Indeed, for $d=1$ graphs, the state evolution is that of the Grover algorithm, obtaining $1 - \mathcal{O}(1/N)$ success probability and having optimal query complexity approximately $\frac{\pi}{4}\sqrt{N}$.

Furthermore, using techniques similar to that used to eliminate the $\mathcal{O}(1/N)$ error from Grover search \cite{Long2001,Brassard2002,Toyama2012}, we provide a way to eliminate error from a component of the alternating phase-walk state evolution without affecting the efficiency of the algorithm. We utilise this modification to prove $\mathcal{O}(1)$ success probability on any $d=2$ graph. In combination with our recent work \cite{marsh2021}, there is strong evidence that the state evolution can be made entirely deterministic to achieve an exact unitary evolution from the equal superposition to a marked vertex on any periodic graph.

Finally, we apply the algorithm to some example graph classes. Most notably, we study a specific class of rook graphs that represent the valid moves of a rook piece on a rectangular $n \times n^3$ chessboard. The \cg algorithm is known to perform suboptimally on this class of graphs, taking time $\bigo{N^{5/8}}$ to achieve $\mathcal{O}(N^{-1/8})$ success probability~\cite{Chakraborty2020}. In contrast, the alternating phase-walk framework yields $\bigo 1$ overlap with $p=\bigo{\sqrt{N}}$ on any rook graph for which $d=2$, achieving asymptotically optimal search.

\section{Search framework}

\subsection{Motivation}

As a starting point, in this section we interpret Grover's algorithm as an alternating phase-walk spatial search on the complete graph in order to motivate a generalisation. The connection of Grover search to the complete graph is well-known in the context of the \cg algorithm \cite{Farhi1998}. Moreover, \cg spatial search is a generalisation of the `analog analogue' of Grover's search from the complete graph to an arbitrary graph Laplacian $\mathcal{L}$. Our framework can be thought of as the same generalisation applied to the Grover iterate. We display this relationship in \cref{tab:searchalgs}. Note that in this table we express the Grover iterate in its generalised form, which is used to achieve deterministic search \cite{Grover1996,Long2001,Brassard2002,Hoyer2000}, and also for other purposes such as optimal fixed-point search \cite{Yoder2014}.

\begin{table}[h]
    \centering
    \begin{tabular}{|l|c|}\hline
        \textbf{Search framework} & \textbf{Search unitary} \\ \hline
        Grover \cite{Grover1996,Long2001,Brassard2002,Hoyer2000,Yoder2014} & $U_k = e^{-i \beta_k \ket{s} \bra{s}}e^{-i \alpha_k \ket{\omega} \bra{\omega}}$ \\ \hline
        Analog Grover \cite{Farhi1998} & $U = e^{-i t (\ket{s} \bra{s} + \ket{\omega} \bra{\omega})}$ \\ \hline
        Spatial ($\mathcal{CG}$) \cite{Childs2004} & $U = e^{-i t (\gamma \mathcal{L} + \ket{\omega} \bra{\omega})}$ \\ \hline
        Alternating phase-walk & $U_k = e^{-i t_k \mathcal{L}} e^{-i \theta_k \ket{\omega} \bra{\omega}}$ \\ \hline
    \end{tabular}
    \caption{Comparison between the ansatzes of quantum search algorithms.}
    \label{tab:searchalgs}
\end{table}

Explicitly, the original Grover iterate for an $N$-element database can be expressed as
\begin{equation}
	U = \uw{\frac{\pi}{N}} \uf{\pi} \, .
\end{equation}
Here, $\uf{\pi}$ performs a $\pi$-phase shift of the marked element,
\begin{equation}
	\uf{\pi} = e^{-i \pi \ket{\omega} \bra{\omega}} = \mathbb{I} - 2 \ket{\omega}\bra{\omega} \, ,
\end{equation}
and $\uw{\frac{\pi}{N}}$ is a quantum walk on the complete graph $\mathbb{K}_N$ for time $t=\frac{\pi}{N}$. The Laplacian for the complete graph can be expressed as $\mathcal{L} = N (\mathbb{I} - \ket{s} \bra{s})$, so
\begin{equation}
	\uw{\frac{\pi}{N}} = e^{-i t \mathcal{L}} = e^{-i \frac{\pi}{N} (N \mathbb{I} - N \ket{s} \bra{s})} = -(\mathbb{I} - 2\ket{s}\bra{s}) \, .
	\label{eq:completegraphwalk}
\end{equation}
The above is recognisable as the Grover diffusion operator \cite{Grover1996}. The iterate $U$ then takes the familiar form
\begin{equation}
    U = -(\mathbb{I} - 2 \ket{s} \bra{s}) (\mathbb{I} - 2 \ket{\omega} \bra{\omega})
\end{equation}
where evolution is restricted to a 2-dimensional subspace spanned by the marked vertex and the equal superposition. It is well known \cite{Grover1996,Brassard2002} that
\begin{equation}
    \ket{\omega} = U^{(p-1)/2} \ket{s}
    \label{eq:grover}
\end{equation}
for $p=\frac{\pi}{2\arccos\sqrt{1 - N^{-1}}}$. Thus for large $N$, approximately $\frac{\pi}{4}\sqrt{N}$ iterations (or equivalently, oracle queries) are used to map the equal superposition to the marked element. The query complexity of Grover's algorithm is proven to be the exact lower limit for black-box quantum search \cite{Zalka1999,Dohotaru2009}. Hence we can conclude that the complete graph admits exactly optimal spatial search under the alternating phase-walk formalism.

However, what if the Laplacian of the complete graph is replaced with that of another graph? Clearly if a certain $t$ can be found such that $U_w(t) = e^{-i t \mathcal{L}} = -(\mathbb{I} - 2\ket{s}\bra{s})$, the search iterate will be functionally identical to Grover's algorithm and thus will admit optimal quantum search. To determine when this is the case, we write the walk operator in terms of the spectral decomposition of the Laplacian. Denote the set of Laplacian eigenvalues and eigenstates as $\Lambda=\{ \lambda_0, \ldots, \lambda_{N-1} \}$ and $\chi = \{ \ket{b_0}, \ldots, \ket{b_{N-1}} \}$ respectively. Note that the equal superposition $\ket{s} \equiv \ket{b_0}$ is always an eigenstate of $\mathcal{L}$ with eigenvalue $\lambda_0 = 0$, since the rows and columns of a Laplacian matrix sum to zero. We will also make the assumption that $\braket{\omega}{b} \geq 0$ for all $\ket{b} \in \chi$ without loss of generality: if not, set $\ket{b} \vcentcolon= -\ket{b}$. Re-writing the quantum walk operator,
\begin{align}
	U_w(t) &= e^{-i t \mathcal{L}} = \sum\limits_{k=0}^{N-1} e^{-i t \lambda_k} \ket{b_k} \bra{b_k} \\
	&= -\left(\mathbb{I} - 2 \ket{s} \bra{s} - \sum\limits_{k=1}^{N-1} (1 + e^{-i t \lambda_k}) \ket{b_k} \bra{b_k}\right) \, ,
	\label{eq:walkdualbasis}
\end{align}
where we have used $\mathbb{I} = \sum\limits_{k=0}^{N-1} \ket{b_k} \bra{b_k}$.

Thus, if there exists a $t$ such that for all $1 \leq k \leq N-1$, $e^{-i t \lambda_k}=-1$, then \cref{eq:walkdualbasis} reduces to \cref{eq:completegraphwalk} and we get the optimal behaviour exhibited by the complete graph. With this in mind, a sensible first walk time is $t = \frac{\pi}{\gcd \Lambda}$. For each $k$, if $\frac{\lambda_k}{\gcd \Lambda}$ is odd then the corresponding term will be eliminated. Utilising the greatest common divisor requires that all eigenvalues are rational, thereby restricting to periodic graphs \cite{Godsil2011}. In \cref{sec:phasewalkapplications} we give an example class of graphs, separate from the complete graph, where this single choice of walk time $t$ produces a state evolution that matches Grover's algorithm. For most graphs, however, clearly this single walk time will not suffice to eliminate the interaction from other Laplacian eigenstates. In other words, a single fixed $t$ does not eliminate \textit{all} of the extra terms from \cref{eq:walkdualbasis}. We will instead need a series of walk times, interleaved with marked vertex phase shifts, that iteratively eliminate the interaction with other Laplacian eigenstates. This is the core idea behind the spatial search algorithm proposed in the following section.

\subsection{Spatial search on periodic graphs}
\label{sec:spatialsearchgeneral}

We start by making a series of definitions that will be necessary for the proposed quantum spatial search. First, define the initial set of unique non-zero Laplacian eigenvalues $\Lambda_0 = \Lambda\setminus\{ 0\}$. Then iteratively for $k \geq 1$, let
\begin{align}
    t_{k} &= \frac{\pi}{\gcd \Lambda_{k-1}} \, , \\
    \Lambda_k &= \{ \lambda \in \Lambda_{k-1} \mid  e^{-i t_k \lambda} = 1 \} \, \\
    \bar{\Lambda}_k &= \{ \lambda \in \Lambda_{k-1} \mid  e^{-i t_k \lambda} = -1 \} \, .
\end{align}
That is, for each $k$, we define a quantum walk time $t_{k}$ depending on the greatest common divisor of a subset of the Laplacian's eigenvalues. We make note that the use of the $\gcd$ on subsets of the graph spectrum also arises in determining the minimum period $\tau$ of a periodic vertex~\cite{Godsil2011}. The walk time subdivides the current set of eigenvalues $\Lambda_{k-1}=\Lambda_k \cup \bar{\Lambda}_k$ into two disjoint sets. For any $k$ we must have at least one element in $\bar{\Lambda}_k$, since otherwise all elements in $\Lambda_{k-1}$ would have a common factor of 2, contradicting the $\gcd$. Hence, by continual subdivision $\Lambda_{k-1}\xrightarrow{t_k} \{\Lambda_k, \bar{\Lambda}_k\}$, eventually $\Lambda_{d} = \emptyset$ for  $d \geq 1$ divisions, and the above formulation produces a finite set of $d$ walk times $t_1, \ldots, t_d$.

Furthermore, we also define for each $0 \leq k \leq d$ the corresponding sets of Laplacian eigenstates $\chi_k$ and $\bar{\chi}_k$ (including all degenerate eigenstates). We associate with each of these sets the projection of the marked vertex $\ket{\omega}$ onto the elements,
\begin{align}
    \ket{\chi_k} &= \frac{1}{\sqrt{\sum\limits_{\ket{b} \in \chi_k} \braket{\omega}{b}^2}}\sum\limits_{\ket{b} \in \chi_k} \braket{\omega}{b} \ket{b} \, ,\\
    \ket{\bar{\chi}_k} &= \frac{1}{\sqrt{\sum\limits_{\ket{b} \in \bar{\chi}_k} \braket{\omega}{b}^2}}\sum\limits_{\ket{b} \in \bar{\chi}_k} \braket{\omega}{b} \ket{b} \, .
    \label{eq:gendef}
\end{align}
With these definitions made, we now give the main result of this paper.
\begin{theorem}[Spatial search on periodic graphs]
The unitary evolution
\begin{equation}
    \ket{\omega} = U_1^{(p_1-1)/2} U_2^{(p_2-1)/2} \ldots U_d^{(p_d-1)/2} \ket{s}
\end{equation}
maps the equal superposition $\ket{s}$ to a marked vertex $\ket{\omega}$, where
\begin{align}
    U_0 &= U_f(\pi) \, , \\
    U_k &= U_w(t_k) (U_{k-1})^{p_{k-1}} \, ,\\
    p_0 &= 1 \, , \\
    p_k &= \frac{\pi}{2\arccos{\frac{\braket{\omega}{\bar{\chi}_k}}{\sqrt{\braket{\omega}{s}^2 + \braket{\omega}{\chi_{k-1}}^2}}}} \qquad 1 \leq k \leq d \, . \label{eq:pk}
\end{align}
\label{thm:genproof}
\end{theorem}
Before proceeding to prove \cref{thm:genproof}, we make some preliminary comments on the state evolution. For vertex-transitive graphs we have $\braket{\omega}{\chi_k} = \sqrt{\sum_{\lambda \in \Lambda_k} \mu_\lambda/N}$ and $\braket{\omega}{\bar{\chi}_k} = \sqrt{\sum_{\lambda \in \bar{\Lambda}_k} \mu_\lambda/N}$, where $\mu_\lambda$ is the algebraic multiplicity of eigenvalue $\lambda$. This makes the state evolution independent of the specific marked element, as would be expected from vertex-transitivity.

The parameter $d$, which affects the complexity of the state evolution, is bounded above by the number of unique non-zero eigenvalues. A `worst-case' example is $\Lambda_0=\{1, 2, 4, \ldots, 2^{l-1}\}$, resulting in $d=l$. However it is often significantly less than the upper bound, such as $\Lambda_0=\{ 1, 2, \ldots, (l-1)\}$ resulting in $d=\ceil{\log_2 l}$. Of particular note are integral graphs, which are defined by an integer-valued spectrum and have previously arisen in the studies of quantum perfect state transfer \cite{Harary1974,Omran2009}. For the integer spectrum case, $d$ is equal to the number of unique exponents of 2 in the prime factorisations of the non-zero eigenvalues.

The total number of iterations is dependent on the iteration counts $p_k$, where we have $p_k \geq 1$ for all $1 \leq k \leq d$. By comparison to the motivating Grover case given in \cref{eq:grover}, it is clear that  $d=1$ graphs achieve optimal quantum search, matching the query complexity of Grover's algorithm. More generally, the total number of phase-walk iterations (or equivalently, oracle queries) for the state evolution is
\begin{equation}
    N_\text{iter} = \frac{1}{2} \sum\limits_{k=1}^d (p_k-1) \prod\limits_{j=1}^{k-1} p_j = \frac{1}{2}\left(-1 + \prod\limits_{k=1}^d p_k \right) \, .
    \label{eq:niter}
\end{equation}
We will later show that for fixed $d$, $N_\text{iter} = \mathcal{O}(\sqrt{N})$.

% In the following section, we prove \cref{thm:genproof}. We then show that in the $d=1$ case, the Grover state evolution is obtained exactly. We then study the $d=2$ case, showing that asymptotically optimal quantum spatial search is always achieved. Finally, we explore example classes of graphs that exhibit the $d=1, 2, 3$ cases, with each achieving asymptotically optimal quantum spatial search.

% Thus, we have established a general spatial search algorithm for periodic graphs. The walk times $t_k$ are easily derived from the spectrum of the Laplacian, and similarly the iterations $p_k$ and $p_k'$ are straightforward to calculate directly from the Laplacian eigenstates. The issue for $d \geq 2$ in general, however, is rounding of these iteration counts $p_k$ and $p_k'$ to the nearest integer for practical implementation. Depending on the graph, compounding error from sub-iterates has the potential to eradicate the final $\mathcal{O}(1)$ overlap with the marked vertex. However, there is strong numerical evidence to indicate that by generalising $U_f(\pi) \mapsto U_f(\theta)$, one can not only improve the success probability, but achieve fully deterministic quantum search. We give an example of such a case in [cite previous paper]. Additionally, in the following section, we examine an interesting $d=2$ case where $\mathcal{O}(1)$ search can be achieved using these generalised phase rotations. 

\subsection{Proof of correctness}

In order to prove \cref{thm:genproof}, it is useful to first define the component of the marked vertex $\ket{\omega}$ projected onto $\ket{s}$ and $\ket{\chi_k}$,
\begin{equation}
    \ket{\omega_k} = \frac{1}{\sqrt{\braket{\omega}{s}^2 + \braket{\omega}{\chi_k}^2}}\left(\braket{\omega}{s} \ket{s} + \braket{\omega}{\chi_k} \ket{\chi_k}\right) \, .
\end{equation}
Observe that $\ket{\omega_d} = \ket{s}$ (since $\chi_d=\emptyset$), and $\ket{\omega_0} = \ket{\omega}$ (since $\chi_0$ is the set of all eigenstates except the equal superposition). In this section, we will show that the state evolution in \cref{thm:genproof} maps $\ket{\omega_d} \mapsto \ket{\omega_{d-1}} \mapsto \ldots \mapsto \ket{\omega_0}$, and thus evolves from the equal superposition to the marked vertex.

With this definition in hand, the following lemma describes the form of each iterate $U_k$.
\begin{lemma}
Restricted to the subspace spanned by $\{ \ket{s}, \ket{\chi_{k}}, \ket{\bar{\chi}_{k}}\}$, the iterate $U_k$ can be expressed as the composition of two Householder reflections
\begin{equation}
    U_k = (\mathbb{I} - 2 \ket{\bar{\chi}_k} \bra{\bar{\chi}_k})(\mathbb{I} - 2 \ket{\omega_{k-1}}\bra{\omega_{k-1}}) \, ,
\end{equation}
for $1 \leq k \leq d$.
\label{lem:lemma1}
\end{lemma}
\begin{proof}
We prove by induction. First consider the base case, where $k=1$. By definition, $U_1 = \uw{t_1} \uf{\pi}$. We have $\ket{\omega} = \ket{\omega_0}$. Then clearly
\begin{equation}
    \uf{\pi} = \mathbb{I} - 2 \ket{\omega} \bra{\omega} = \mathbb{I} - 2 \ket{\omega_0} \bra{\omega_0} \, .
\end{equation}
Following from the definition of the walk time $t_1$ and \cref{eq:gendef}, where the phases of all Laplacian eigenstates $\ket{b} \in \bar{\chi}_1$ are flipped,
\begin{equation}
    \uw{t_1} = \mathbb{I} - 2 \ket{\bar{\chi}_1} \bra{\bar{\chi}_1}
\end{equation}
as required, fulfilling the base case.

For the inductive step, it is convenient to express $U_k$ as a matrix in the basis $\{ \ket{s}, \ket{\chi_k}, \ket{\bar{\chi}_k}\}$,
\begin{widetext}
\begin{align}
    U_k = \begin{pmatrix}
    1 - \frac{2 \braket{\omega}{s}^2}{\braket{\omega}{s}^2 + \braket{\omega}{\chi_{k-1}}^2} & -\frac{2 \braket{\omega}{\chi_k} \braket{\omega}{s}}{\braket{\omega}{s}^2 + \braket{\omega}{\chi_{k-1}}^2} & -\frac{2 \braket{\omega}{\bar{\chi}_k} \braket{\omega}{s}}{\braket{\omega}{s}^2 + \braket{\omega}{\chi_{k-1}}^2} \\
    -\frac{2 \braket{\omega}{\chi_k} \braket{\omega}{s}}{\braket{\omega}{s}^2 + \braket{\omega}{\chi_{k-1}}^2} & 1 - \frac{2 \braket{\omega}{\chi_k}^2}{\braket{\omega}{s}^2 + \braket{\omega}{\chi_{k-1}}^2} & -\frac{2 \braket{\omega}{\chi_k} \braket{\omega}{\bar{\chi}_k}}{\braket{\omega}{s}^2 + \braket{\omega}{\chi_{k-1}}^2} \\
    \frac{2 \braket{\omega}{\bar{\chi}_k} \braket{\omega}{s}}{\braket{\omega}{s}^2 + \braket{\omega}{\chi_{k-1}}^2} & \frac{2 \braket{\omega}{\chi_k} \braket{\omega}{\bar{\chi}_k}}{\braket{\omega}{s}^2 + \braket{\omega}{\chi_{k-1}}^2} & -1 + \frac{2 \braket{\omega}{\bar{\chi}_k}^2}{\braket{\omega}{s}^2 + \braket{\omega}{\chi_{k-1}}^2}
    \end{pmatrix} \, .
    \label{eq:bigmatrix}
\end{align}
\end{widetext}
This operator has eigenphases $0$ and $\pm \lambda^{(k)}$, where
\begin{align}
    \lambda^{(k)} &= 2 \arccos\frac{\braket{\omega}{\bar{\chi}_k}}{\sqrt{\braket{\omega}{s}^2 + \braket{\omega}{\chi_{k-1}}^2}}
\end{align}
notably appears as the denominator of $p_k$ in \cref{thm:genproof}.
The corresponding eigenstates are
\begin{align}
    \ket{v_0^{(k)}} &= \frac{1}{\sqrt{\braket{\omega}{s}^2 + \braket{\omega}{\chi_k}^2}}\left(\braket{\omega}{\chi_k} \ket{s} - \braket{\omega}{s} \ket{\chi_k}\right) \, , \\
    \ket{v_\pm^{(k)}} &= \frac{1}{\sqrt{2}} \left( \ket{\bar{\chi}_k} \pm i \ket{\omega_k} \right) \, .
    \label{eq:geneigstate}
\end{align}
The eigensystem can be used to diagonalise $U_k$ and compute $U_k^{p_k}$ with $p_k = \frac{\pi}{\lambda^{(k)}}$, resulting in
\begin{align}
    U_k^{p_k} &= \begin{pmatrix}
 1-\frac{2 \braket{\omega}{s}^2}{\braket{\omega}{s}^2+\braket{\omega}{\chi_k}^2} & -\frac{2 \braket{\omega}{s} \braket{\omega}{\chi_k}}{\braket{\omega}{s}^2+\braket{\omega}{\chi_k}^2} & 0 \\
 -\frac{2 \braket{\omega}{s} \braket{\omega}{\chi_k}}{\braket{\omega}{s}^2+\braket{\omega}{\chi_k}^2} & -1 + \frac{2 \braket{\omega}{s}^2}{\braket{\omega}{s}^2+\braket{\omega}{\chi_k}^2} & 0 \\
 0 & 0 & -1 
\end{pmatrix} \\
&= \mathbb{I} - 2 \ket{\omega_k} \bra{\omega_k} \, .
\label{eq:upowerpk}
\end{align}
Therefore by performing $p_k$ iterations of $U_k$, we create a Householder reflection about the $\ket{\omega_k}$ state. Furthermore recall that $\chi_k = \chi_{k+1} \cup \bar{\chi}_{k+1}$, and so restricted to the basis $\{ \ket{s}, \ket{\chi_{k+1}}, \ket{\bar{\chi}_{k+1}}\}$,
\begin{equation}
    \uw{t_{k+1}}=\mathbb{I} - 2\ket{\chi_{k+1}} \bra{\chi_{k+1}} \, .
\end{equation}
Thus, we have
\begin{align}
    U_{k+1} &= \uw{t_{k+1}} U_k^{p_k} \\
    &= (\mathbb{I} - 2\ket{\chi_{k+1}} \bra{\chi_{k+1}}) (\mathbb{I} - 2 \ket{\omega_k} \bra{\omega_k})
\end{align}
as required, proving \cref{lem:lemma1}.
\end{proof}

Having presented the general form of $U_k$ for $1 \leq k \leq d$, we now demonstrate its utility.
\begin{lemma}
For $1 \leq k \leq d$,
\begin{equation}
    U_k^{(p_k-1)/2} \ket{\omega_k} = \ket{\omega_{k-1}} \, .
\end{equation}
\label{lem:lemma2}
\end{lemma}
\begin{proof}
By observation of the $\ket{v_\pm^{(k)}}$ eigenstate, clearly $U_k$ induces a rotation between $\ket{\omega_k}$ and $\ket{\bar{\chi}_k}$. Specifically, for any $p$,
\begin{equation}
    U_k^p \ket{\omega_k} = \cos p \lambda_+^{(k)} \ket{\omega_k} + \sin p \lambda_+^{(k)} \ket{\bar{\chi}_k} \, .
    \label{eq:genericpower}
\end{equation}
By choosing
\begin{align}
    p&=\frac{1}{\lambda^{(k)}}\arcsin{\frac{\braket{\omega}{\bar{\chi}_k}}{\sqrt{ \braket{\omega}{s}^2 + \braket{\omega}{\chi_{k-1}}^2}}} \\
    &= \frac{1}{\lambda^{(k)}}\left( \frac{\pi}{2} - \frac{\lambda^{(k)}}{2}\right) = \frac{1}{2}(p_k-1) \, ,
\end{align}
we obtain
\begin{equation}
    (U_k)^{(p_k-1)/2} \ket{\omega_k} = \ket{\omega_{k-1}} \, .
\end{equation}
Here we have again used the property that $\bar{\chi}_{k-1} = \chi_{k} \cup \bar{\chi}_{k}$, and consequently $\braket{\omega}{\chi_k} \ket{\chi_k} + \braket{\omega}{\chi_k} \ket{\bar{\chi}_k} = \braket{\omega}{\chi_{k-1}} \ket{\chi_{k-1}}$. Thus, \cref{lem:lemma2} holds.
\end{proof}

\cref{thm:genproof} directly follows from mapping $\ket{\omega_d} \mapsto \ket{\omega_{d-1}} \mapsto \ldots \mapsto \ket{\omega_0}$, where $\ket{\omega_d} = \ket{s}$ and $\ket{\omega_0} = \ket{\omega}$. Thus we see that on a periodic graph, one can compose quantum walks and marked vertex phase shifts to create Householder reflections that iteratively rotate the equal superposition to the marked vertex.

\subsection{Time complexity}
\label{sec:efficiency}
    
In this section, we discuss the efficiency of the state evolution. The most sensible measure is the number of calls to $U_f(\theta)$ (the query complexity), or equivalently the total number of phase-walk iterations $N_\text{iter}$. This value is directly comparable to the evolution time used by the \cg algorithm, since simulating the search Hamiltonian for time $T$ is equivalent to performing $\mathcal{O}(T)$ oracle queries \cite{Roland2003}. Thus, in the following we determine the worst-case query complexity $N_\text{iter}$ in terms of the graph parameter $d$ and the number of vertices $N$.

The expression for $N_\text{iter}$ given in \cref{eq:niter} has $d$ independent variables $N$ and $\braket{\omega}{\chi_k}$ for $k=1, \ldots, d-1$. Re-writing in terms of these variables using $\braket{\omega}{\bar{\chi}_k} = \sqrt{ \braket{\omega}{\chi_{k-1}}^2 - \braket{\omega}{\chi_k}^2}$ and $\braket{\omega}{s} = 1/\sqrt{N}$,
\begin{equation}
    p_k = \frac{\pi}{2 \arccos{\sqrt{\frac{\braket{\omega}{\chi_{k-1}}^2 - \braket{\omega}{\chi_k}^2}{1/N + \braket{\omega}{\chi_{k-1}}^2}}}} \, .
\end{equation}
The partial derivative $\frac{\partial N_\text{iter}}{\partial \braket{\omega}{\chi_j}}$ involves two terms from the product in \cref{eq:niter}. Solving for $\frac{\partial N_\text{iter}}{\partial \braket{\omega}{\chi_j}}=0$ produces the set of equations
\begin{equation}
    p_j \frac{\partial p_{j+1}}{\partial \braket{\omega}{\chi_j}} + p_{j+1} \frac{\partial p_j}{\partial \braket{\omega}{\chi_j}} = 0
\end{equation}
for $1 \leq j \leq {d-1}$. Evaluating the derivative and simplifying results in
\begin{equation}
    p_j \tan \frac{\pi}{2 p_j} = p_{j+1} \tan \frac{\pi}{2 p_{j+1}} \, .
\end{equation}
Thus, it follows that the extremum of $N_\text{iter}$ occurs when $p_1 = p_2 = \ldots = p_d = p$, producing
\begin{equation}
    p = \frac{\pi}{2\arccos{\sqrt{1 - N^{-1/d}}}} \, .
\end{equation}
Substituting and evaluating the series expansion of $N_\text{iter}$ for large $N$,
\begin{align}
    N_\text{iter} &= \frac{1}{2}\left(-1 + \left(\frac{\pi}{2}\right)^d\frac{1}{\arccos^{d}\sqrt{1 - N^{-1/d}}}\right) \\
    &= \frac{1}{2}\left(\frac{\pi}{2}\right)^d \sqrt{N} - \frac{1}{2} - \mathcal{O}(d \, N^{1/2 - 1/d}) \, .
\end{align}
Hence, asymptotically optimal quantum search is always achieved on periodic graphs when $d$ is independent of $N$. Recall that $d$ is bounded above by the number of distinct eigenvalues. Furthermore, one can obtain optimal search up to polylogarithmic factors if $d$ grows sufficiently slowly with respect to $N$. As a significant example of the latter case, the $n$-dimensional hypercube $\mathbb{Q}_n$ with $N=2^n$ vertices has distinct eigenvalues $\Lambda = \{ 0, 2, \ldots, 2n \}$. For this class of graphs, $d=\ceil{\log_2(n+1)} = \mathcal{O}(\log\log N)$ and therefore $N_\text{iter}=\mathcal{O}(\sqrt{N} \log N)$.

\begin{figure}
    \centering
    \includegraphics[width=\linewidth]{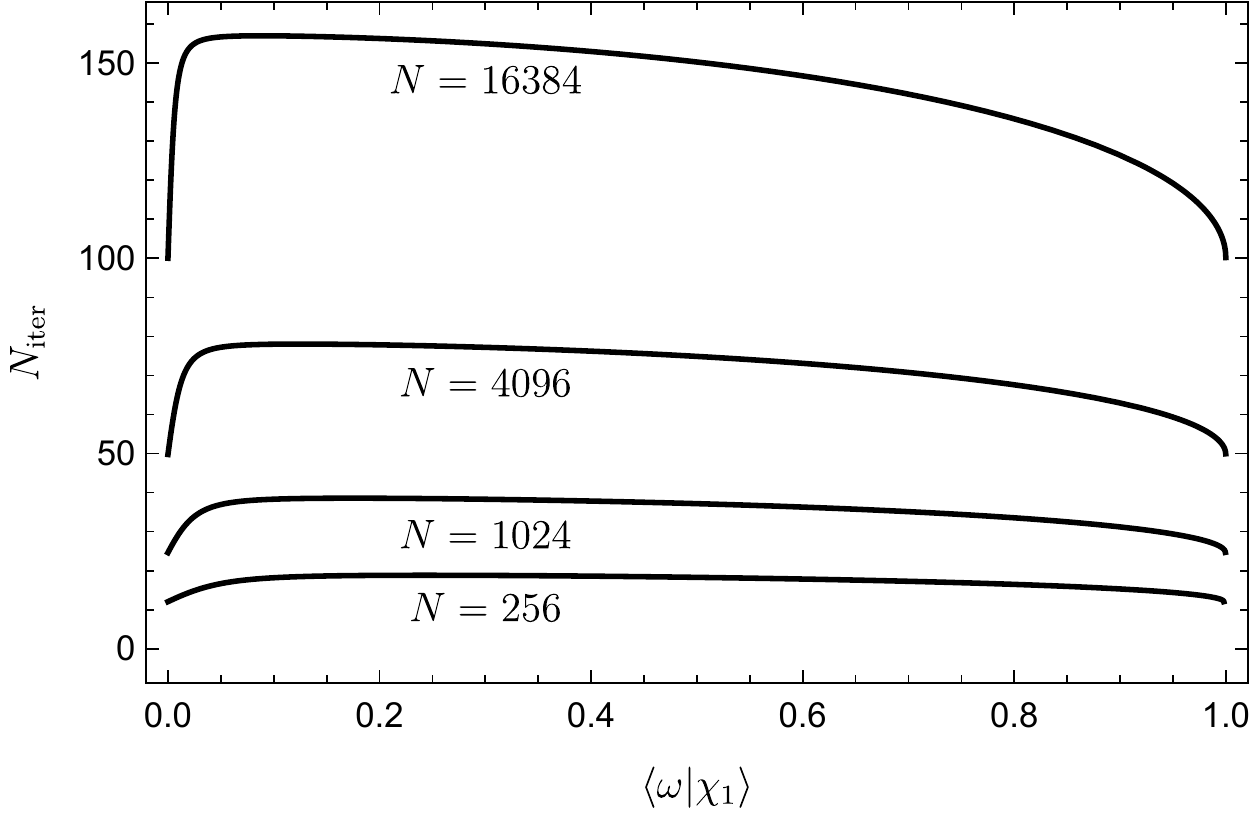}
    \caption{The total number of phase-walk iterations for the $d=2$ case, depending on the marked element overlap with $\ket{\chi_1}$ as per \cref{eq:pk} and \cref{eq:niter}.}
    \label{fig:overlap-vs-queries}
\end{figure}

% This result on the hypercube is consistent with prior research by \citet{Jiang2017} where the Grover diffusion operator is replaced with a transverse field circuit, which is equivalent to a quantum walk on the hypercube. The authors propose a state evolution using $\mathcal{O}(\sqrt{N})$ iterations to reach a maximum of 50\% probability. On average, a 50\% success-rate algorithm will have to be repeated $\mathcal{O}(\log N)$ times, making the overall query complexity $\mathcal{O}(\sqrt{N} \log N)$.

To reinforce the results obtained in this section, we plot the total iterations used for any $d=2$ graph in \cref{fig:overlap-vs-queries}, for $0 \leq \braket{\omega}{\chi_1} \leq \sqrt{1 - N^{-1}}$. For the boundary cases $\braket{\omega}{\chi_1}= 0$ and $\braket{\omega}{\chi_1}= \sqrt{1-N^{-1}}$, we see that $N_\text{iter}$ reduces to the $d=1$ case. The maximum, which occurs when $p_1 = p_2$, results in approximately $\frac{\pi}{2}$ more iterations than the $d=1$ Grover case.

\subsection{Controlling the rotation rate}
\label{sec:exactu1}

Although the state evolution given in \cref{thm:genproof} holds exactly for any periodic graph, issues can arise in practice due to the requirement that the number of iterations must be integer-valued. For each $k$, rounding $p_k$ and $\frac{1}{2}(p_k-1)$ to the nearest integers introduces error, and compounding error from sub-iterates inside sub-iterates can eliminate the desired $\mathcal{O}(1)$ overlap with the marked vertex, especially for higher $d$. However, in this section we propose a way to remove this compounding error from $U_1^{p_1}$ for $d=2$, motivating research into the general case. To be explicit, the $d=2$ state evolution can be written as
\begin{align}
    \ket{\omega} &= U_1^{(p_1-1)/2} U_2^{(p_2-1)/2} \ket{s} \\
    &= U_1^{(p_1-1)/2} \left(\uw{t_2} U_1^{p_1} \right)^{(p_2-1)/2} \ket{s} \, ,
\end{align}
where $U_1 = \uw{t_1} \uf{\pi}$.
We provide a technique to implement the $U_1^{p_1}$ component exactly. This technique is inspired by the approach described in \cite{Long2001} to make the Grover state evolution deterministic. Consider the generalised marked element phase rotation
\begin{align}
    U_f(\theta) = e^{-i \theta \ket{\omega} \bra{\omega}} = \mathbb{I} + (-1 + e^{i \theta}) \ket{\omega} \bra{\omega} \, .
\end{align}
The key observation is that generalising $U_f(\pi) \mapsto U_f(\theta)$ does not affect the 3-dimensional evolution subspace from \cref{lem:lemma1}, i.e.
\begin{equation}
    \left(\prod_j U_w(t_1) U_f(\theta_j)\right) \ket{s} \in \spn\{ \ket{s}, \ket{\chi_1}, \ket{\bar{\chi}_1}\} \, .
\end{equation}
We will utilise this property to eliminate error from $U_1^{p_1}$. Define
\begin{equation}
    U_1(\theta) = U_w(t_1) U_f(-\theta) U_w(t_1) U_f(\theta) \, ,
\end{equation}
which has eigenphases $0$ and $\pm \lambda_\theta^{(1)}$, where
\begin{align}
    \lambda^{(1)}_\theta &= 2 \arcsin\left( \sin\frac{\theta}{2} \sin\frac{\pi}{p_1}\right)
\end{align}
using $p_1 = \frac{\pi}{2\arccos\braket{\omega}{\bar{\chi}_1}}$. Clearly we have $\lambda^{(1)}_\pi =2\lambda^{(1)}$, since $U_1(\pi) = U_1^2$.
The eigenstates of $U_1(\theta)$ (left unnormalised) are
\begin{align}
    \ket{v_0^{(1)}} =& \braket{\omega}{\chi_1} \ket{s} - \braket{\omega}{s} \ket{\chi_1} \, ,\\
    \ket{v_\pm^{(1)}} =& \braket{\omega}{s} \ket{s} + \braket{\omega}{\chi_1} \ket{\chi_1} \\
    &+ e^{-i \arctan\left(\tan\frac{\theta }{2} \cos\frac{\pi }{p_1}\right)} \sin \frac{\pi}{2p_1} \ket{\bar{\chi}_1} \, .
\end{align}

Calculating $U_1(\theta)^{p_1'}$, where $p_1' = \frac{\pi}{\lambda^{(1)}_\theta}$, we obtain an identical unitary to \cref{eq:upowerpk}, with the exception of a phase difference on $\bra{\bar{\chi}_1}U_1^{p_1}\ket{\bar{\chi}_1}$. This phase difference does not affect the evolution, since by construction $\abs{\bra{\bar{\chi}_1} (\uw{t_2} U_1^{p_1})^k \ket{s}} = 0$ for any $k>0$.

With this new parametrised iterate, the necessary number of iterations $p_1'$ is a function of $\theta$, with
\begin{equation}
    p_1' = \frac{\pi}{2 \arcsin\left( \sin\frac{\theta}{2} \sin\frac{\pi}{p_1}\right)} \, .
\end{equation}
Solving for $\theta$ gives
\begin{equation}
    \theta_{p_1'} = 2\arcsin\left(\sin\frac{\pi}{2p_1'} \csc\frac{\pi}{p_1}\right) \, .
\end{equation}
Thus, if $p_1 \geq 2$, setting $p_1' = \ceil{p_1/2}$ suffices to ensure that $\theta_{p_1'}$ is real-valued and we are able to perform $U_1^{p_1}= U_1(\theta_{p_1'})^{p_1'}$ exactly, using an integer number of phase-walk iterations $2\ceil{p_1/2}$.

However, if $1 < p_1 < 2$ then the above technique scales poorly since in order to make $\theta_{p_1'}$ real-valued, we require $p_1' \geq \ceil{\frac{p_1}{2(p_1 - 1)}}$, which diverges for $p_1 \rightarrow 1$. To circumvent this issue with small iteration counts, we can instead directly apply the unitary
\begin{equation}
    U_1^{p_1} = \uw{t_1} \uf{\theta} \uw{t_1} \uf{\phi} \uw{t_1} \uf{\theta} 
\end{equation}
where
\begin{align}
    \theta &= 2\arcsin\left( \frac{1}{2}\csc\frac{\pi}{2p_1}\right) \, , \\
    \phi &= -2\arctan\left( \tan\frac{\theta}{2}\sec\frac{\pi}{p_1}\right)\, .
\end{align}
The relationship between the two angles is comparable to the phase matching condition for Grover-searching with certainty, described in \cite{Hoyer2000}. It is straightforward to verify, by use of the matrix expression in \cref{eq:bigmatrix}, that this series of three phase-walk iterations again acts equivalently to $U_1^{p_1}$ for any iteration count $1 \leq p_1 \leq 3$. The angles are real-valued for all $1 \leq p_1 \leq 3$. Thus, $U_1^{p_1}$ can always be implemented exactly for non-integer $p_1$ without affecting the asymptotic complexity of the state evolution.

We anticipate that this idea can be generalised to all $U_k$ using a series of controlled marked vertex phase shifts $\theta_j$, eliminating the error entirely from the state evolution and achieving deterministic quantum spatial search on any periodic graph. Our preliminary numerical work via optimisation of phase parameters along with prior analytical work support this conjecture. We have shown \textit{fully} deterministic search is shown to be achievable on a $d=2$ class of interdependent networks \cite{marsh2021}. This shows that it is also possible to use generalised phase shifts to additionally implement $U_1^{(p_1-1)/2}$ and $U_2^{(p_2-1)/2}$ exactly, on at least one nontrivial class of graphs.

% It seems likely that generalising to $(\mathbb{I} +(-1 + e^{-i \theta}) \ket{\omega_{k-1}}\bra{\omega_{k-1}})$ in \cref{lem:lemma1} is a pathway to making the $k$th unitary $U_k$ deterministic.

Throughout the following sections of this paper, we will assume the `exact' version of $U_1^{p_1}$ is used to improve the overall success probability.

\subsection{Success probability}
\label{sec:success}

As discussed in the previous section, the fact that error is introduced by rounding $p_k$ and $(p_k-1)/2$ poses a problem. Rather than classifying the graphs for which we have $\mathcal{O}(1)$ success probability at the end of the algorithm, we anticipate that a more appropriate route is to circumvent the discussion entirely by controlling the rate of rotation and making the evolution exact, as discussed in \cref{sec:exactu1}. We leave this improvement to future work, and instead show that by using the exact $U_1^{p_1}$ implementation described in the previous section, one can always obtain $\mathcal{O}(1)$ success probability on any $d=2$ graph.

\begin{figure}
    \centering
    \includegraphics[width=\linewidth]{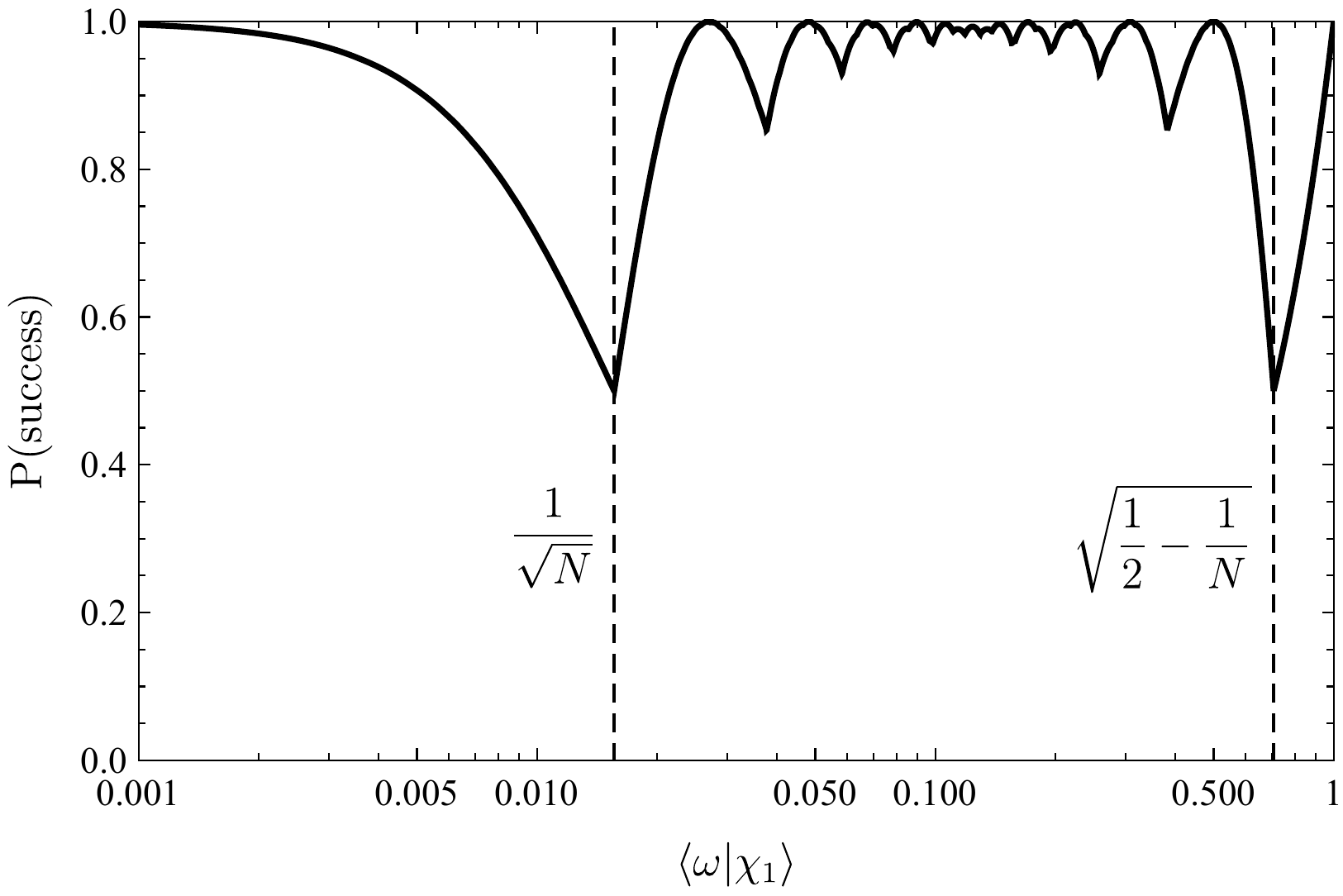}
    \caption{The success probability for the $d=2$ case where $N=4096$, depending on the graph-dependent parameter $\braket{\omega}{\chi_1}$.}
    \label{fig:success}
\end{figure}

We plot the success probability against $\braket{\omega}{\chi_1}$ in \cref{fig:success}, for $N=4096$ vertices. Here all iteration counts are rounded to the nearest integer, and the exact approach to $U_1^{p_1}$ is used. The minimum success probability is approximately 50\%, with two situations producing worst-case error. The first is when $\braket{\omega}{\chi_1} = N^{-1/2}$, resulting in $p_2 = 2 \implies (p_2 - 1)/2 = 1/2$. Clearly this produces 50\% error for implementation of $U_2^{(p_2-1)/2}$, which is the maximum possible. Rounding down to zero iterations gives, calculating directly using \cref{eq:bigmatrix},  
\begin{equation}
    \abs{\bra{\omega} U_1^{\nint{(p_1-1)/2}} \ket{s}}^2 = \frac{1}{2} - \mathcal{O}(N^{-1}) \, .
\end{equation}
The second case is when $\braket{\omega}{\chi_1} = \sqrt{1/2 - N^{-1}}$, resulting in $p_1 = 2 \implies (p_1 - 1)/2 = 1/2$. Again rounding down to zero, we have
\begin{equation}
    \abs{\bra{\omega} U_2^{\nint{(p_2-1)/2}} \ket{s}}^2 = \frac{1}{2} - \mathcal{O}(N^{-1}) \, .
\end{equation}
Thus, we find $\mathcal{O}(1)$ success probability for any value of $\braket{\omega}{\chi_1}$ and consequently asymptotically optimal spatial search for any $d=2$ graph.

\section{Applications}
\label{sec:phasewalkapplications}

\subsection{Johnson graphs}

 The vertices of a Johnson graph $J(n, k)$ correspond to the $k$-element subsets of an $n$-element set, with edges connecting sets that have intersections of size $k-1$. All Johnson graphs are vertex-transitive. When $k=1$, $J(n, 1)$ is exactly the complete graph $\mathbb{K}_n$. In this section we explore the performance of the alternating phase-walk algorithm on $J(n, 2)$ graphs, which will demonstrate both $d=1$ and $d=2$ cases depending on the value of $n$. The $J(n, 2)$ graphs have a total $N = \binom{n}{2}$ vertices. An example of a $k=2$ Johnson graph with 10 vertices is given in \cref{fig:johnson4-2}.

\begin{figure}[b]
    \centering
    \includegraphics[width=0.8\linewidth]{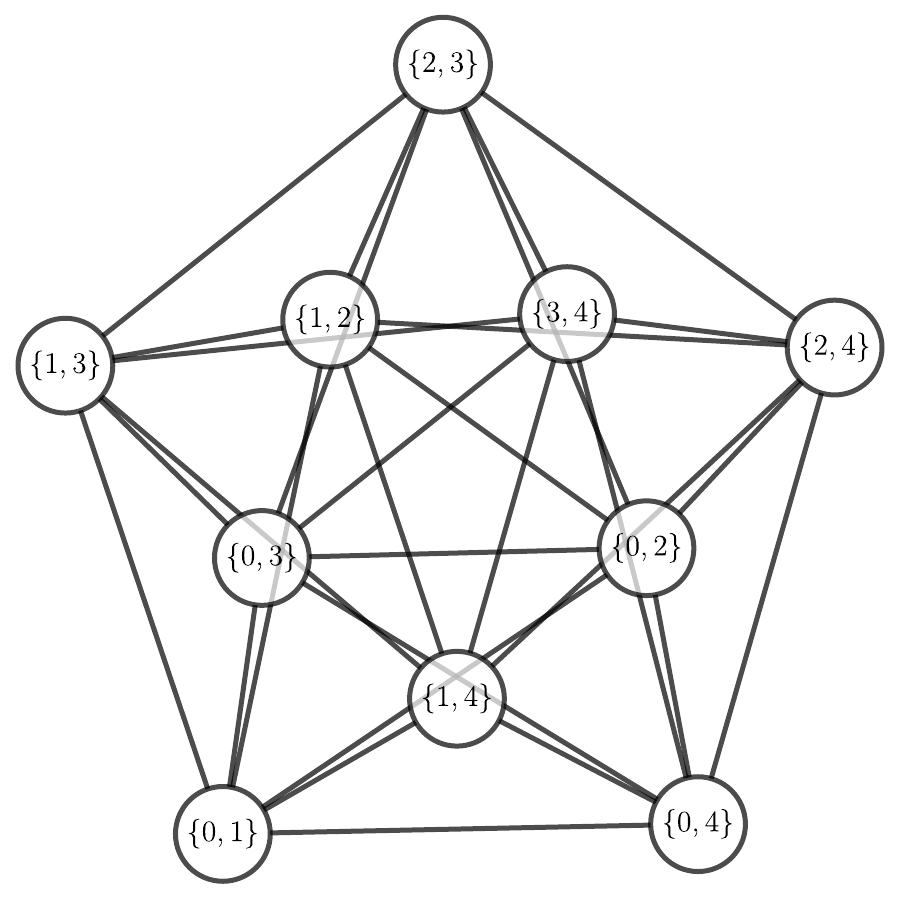}
    \caption{Illustration of a $J(5, 2)$ graph, where edges connect the 2-subsets of $\{ 0, 1, 2, 3, 4 \}$ that share exactly one element.}
    \label{fig:johnson4-2}
\end{figure}

In order to characterise the Laplacian $\mathcal{L}$ for arbitrary $n$ we perform systematic dimensionality reduction using the Lanczos algorithm to the `search subspace' \cite{Novo2015}, which is the Krylov subspace $\mathcal{K}(\mathcal{L}, \ket{\omega})$.  Although this process could be skipped for vertex-transitive graphs by simply directly computing the eigenvalues along with the multiplicities as per \cref{sec:spatialsearchgeneral}, we choose to perform a more detailed analysis in the reduced subspace, which has the additional consequence of eliminating all eigenvalue degeneracies. The dimensionality reduction process is identical to that used for the \cg algorithm \cite{Novo2015}. By vertex-transitivity, for the purposes of analysis the marked vertex can be chosen as any 2-element subset of $[n] = \{0, 1, \ldots (n-1)\}$. Thus, let $\omega = \{ 0, 1 \}$.

The basis generated by the Lanczos algorithm consists of the equal superpositions over subsets that have $(2-j)$ elements in common with the marked element for $j=0, 1, 2$. That is, we have
\begin{align}
    \ket{c_0} &= \ket{\{0, 1\}} \, , \\
    \ket{c_1} &= \frac{1}{\sqrt{2(n-2)}}\sum\limits_{2 \leq x < n} \left(\ket{\{ 0, x \}} + \ket{\{ x, 1 \}}\right) \, , \\
    \ket{c_2} &= \sqrt{\frac{2}{(n-2)(n-3)}} \sum\limits_{\substack{2 \leq x, y < n\\ x\ne y}} \ket{\{x, y\}} \, .
\end{align}
In this basis, the marked vertex is clearly expressed as
\begin{equation}
    \ket{\omega} = (1, 0, 0) \, .
\end{equation}
The equal superposition is given by
\begin{align}
    \ket{s} = \frac{1}{\sqrt{\binom{n}{2}}}\left(1, \sqrt{2(n-2)}, \sqrt{(n-2)(n-3)/2}\right) \, ,
\end{align}
and the Laplacian takes tridiagonal form
\begin{equation}
    \mathcal{L} = \begin{pmatrix}
 2(n-2) & -\sqrt{2(n-2)} & 0 \\
  -\sqrt{2(n-2)} & n-2 & -2 \sqrt{n-3} \\
 0 & -2 \sqrt{n-3} & 4 \\
    \end{pmatrix} \, .
\end{equation}
The non-zero eigenvalues of $\mathcal{L}$ are $\Lambda_0 = \{ n, 2(n-1) \}$, and the eigenstates are
\begin{align}
    \ket{b_0} &= \ket{s} \, ,\\
    \ket{b_1} &= \left(\sqrt{\frac{2}{n}},\frac{n-4}{\sqrt{n(n-2)}},-2\sqrt{\frac{n-3}{n(n-2)}}\right) \, ,\\
    \ket{b_2} &= \left( \sqrt{\frac{n-3}{n-1}},-\sqrt{\frac{2(n-3)}{(n-1) (n-2)}},\sqrt{\frac{2}{(n-1) (n-2)}} \right) \, .
\end{align}
Having obtained the Laplacian eigenvalues and the values of $\braket{\omega}{b_i}$, we now investigate the performance of the algorithm for different $n$. The first walk time $t_1$ is chosen as 
\begin{equation}
    t_1 = \frac{\pi}{\gcd \Lambda_0} = \begin{cases}
        \frac{\pi}{2} & n \equiv 0 \mod 2 \, ,\\
        \pi & n \equiv 1 \mod 2 \, .
    \end{cases}
\end{equation}
It is immediately apparent that the algorithm exhibits different behaviour for even and odd $n$. In fact, there are three different cases, since
\begin{align}
    \uw{\frac{\pi}{2}} =& -(\mathbb{I} - 2\ket{s} \bra{s} - (1+e^{i \frac{n \pi}{2}}) \ket{b_1} \bra{b_1}
    \\&\qquad- (1+e^{i \frac{2(n-1) \pi}{2}}) \ket{b_2} \bra{b_2}) \\
            =& \begin{cases}
                -(\mathbb{I} - 2\ket{s} \bra{s}) & n \equiv 2 \mod 4 \, , \\
                -(\mathbb{I} - 2\ket{s} \bra{s} - 2 \ket{b_1} \bra{b_1}) & n \equiv 0 \mod 4 \, , \\
                -(\mathbb{I} - 2\ket{s} \bra{s} - 2 \ket{b_2} \bra{b_2}) & n \equiv 1 \mod 2 \, .
            \end{cases}
            \label{eq:johnsont1}
\end{align}

\subsubsection{Case 1: $n \equiv 2 \mod 4$}

\begin{figure}
    \centering
    \includegraphics[width=\linewidth]{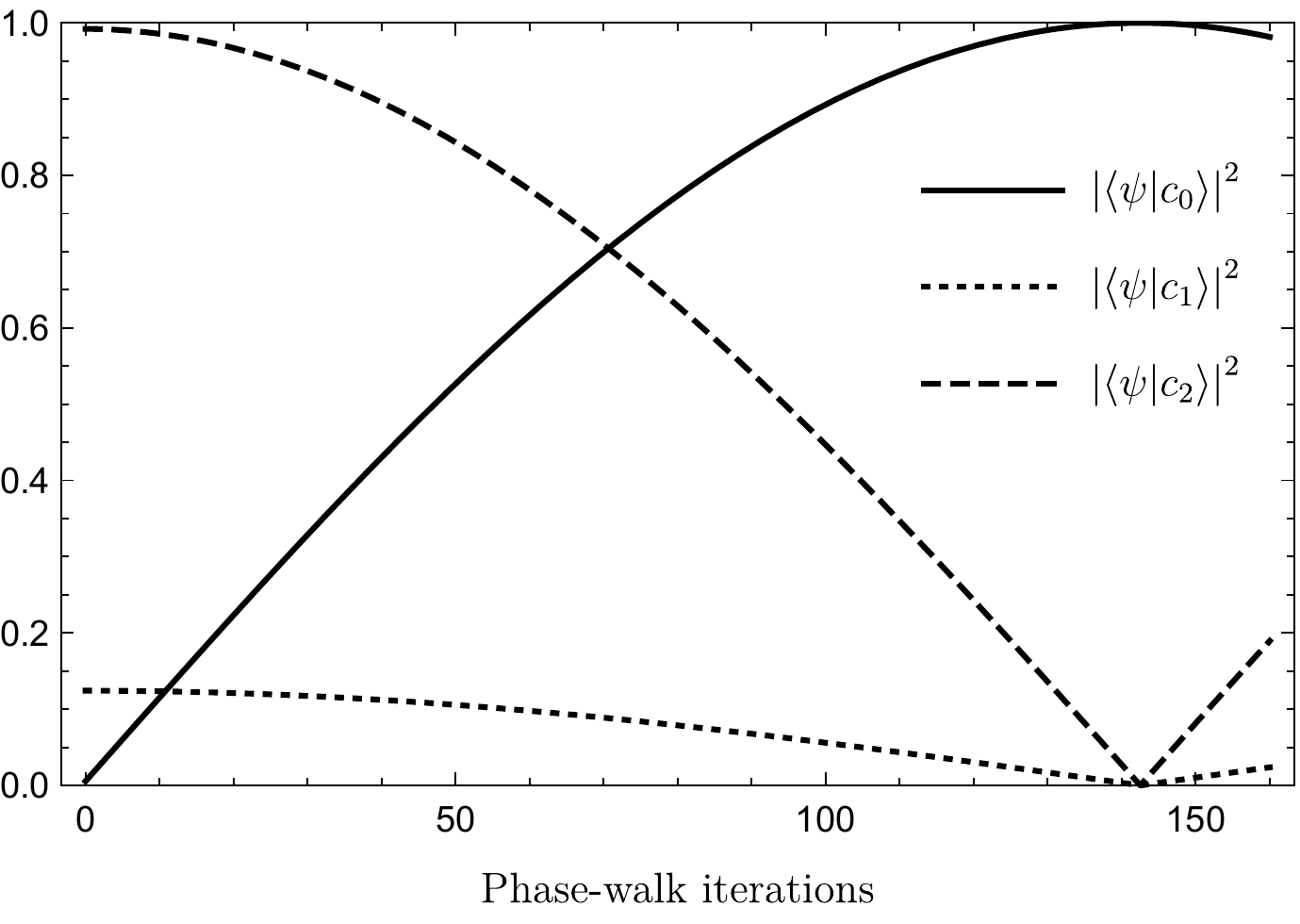}
    \caption{Numerical simulation of the search algorithm for a $J(258, 2)$ graph.}
    \label{fig:johnson2d-258}
\end{figure}

% \begin{figure*}
%     \centering
%     \begin{subfigure}[t]{0.49\textwidth}
%         \centering
%         \includegraphics[width=\linewidth]{Johnson2D-258-Normal.pdf}
%         \caption{}
%     \end{subfigure}
%     \begin{subfigure}[t]{0.49\textwidth}
%         \centering
%         \includegraphics[width=\linewidth]{Johnson2D-258-Eig.pdf}
%         \caption{}
%     \end{subfigure}
%     \caption{Simulation of the search algorithm for a $J(258, 2)$ graph, taking an optimal 142 oracle queries to reach maximum overlap with the marked element. The state evolution is shown in (a) the reduced walk basis, and (b) the Laplacian eigenstate basis.}
%     \label{fig:johnson2d-258}
% \end{figure*}

We first discuss the case where $n \equiv 2 \mod 4$, i.e. the class of graphs $J(4m+2,2)$ for non-negative integers $m$. This is an example of the $d=1$ `Grover case', where the evolution is equivalent to that of the complete graph. That is, with a single walk time $t_1 = \frac{\pi}{2}$, the iterate becomes
\begin{equation}
    U_1 = U_w(\frac{\pi}{2}) \uf{\pi} = -(\mathbb{I} - 2 \ket{s} \bra{s})(\mathbb{I} - 2\ket{\omega} \bra{\omega}) \, .
\end{equation}
After approximately $\frac{\pi}{4}\sqrt{N}$ iterations of $U_w(\frac{\pi}{2}) U_f(\pi)$, the probability of measuring the marked vertex is $1 - \mathcal{O}(\frac{1}{N})$. We plot a numerical simulation of these search dynamics in \cref{fig:johnson2d-258} in the reduced search basis $\{ \ket{c_0}, \ket{c_1}, \ket{c_2}\}$.

\subsubsection{Case 2: $n \equiv 0 \mod 4$}

\begin{figure}
    \centering
    \includegraphics[width=\linewidth]{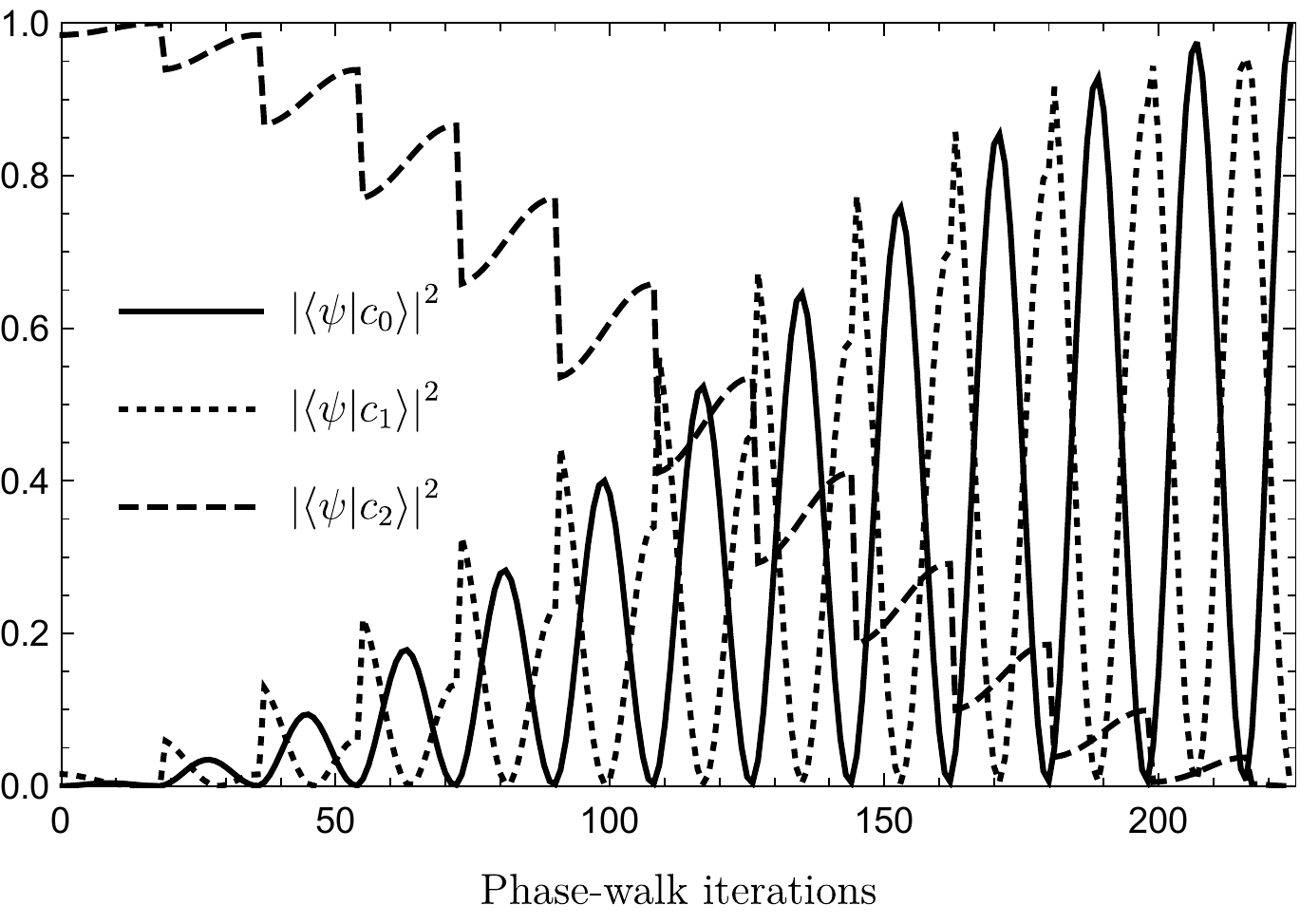}
    \caption{Numerical simulation of the search algorithm for a $J(256, 2)$ graph.}
    \label{fig:approximate-0mod4}
\end{figure}

% \begin{figure*}
%     \centering
%     \begin{subfigure}[t]{0.49\textwidth}
%         \centering
%         \includegraphics[width=\linewidth]{Johnson2D-256-Normal.pdf}
%         \caption{}
%     \end{subfigure}
%     \begin{subfigure}[t]{0.49\textwidth}
%         \centering
%         \includegraphics[width=\linewidth]{Johnson2D-256-Eig.pdf}
%         \caption{}
%     \end{subfigure}
%     \caption{Numerical simulation of the state evolution for a $J(256, 2)$ graph. The overlap is shown in (a) the walk basis, where $\ket{c_0} = \ket{\omega}$, and (b) the Laplacian eigenstate basis, where $\ket{b_0} = \ket{s}$. There is a 99.2\% chance of measuring the marked vertex after 224 oracle queries.}
%     \label{fig:approximate-0mod4}
% \end{figure*}

When $n=4m$ for any integer $m$, the evolution does not reduce to Grover's search and we obtain a $d=2$ case of the framework. By consideration of \cref{eq:johnsont1}, we have $\chi_1 = \{\ket{b_1}\}$ and $\bar{\chi}_1 = \{\ket{b_2}\}$. Hence,
\begin{align}
    p_1 &= \frac{\pi}{2 \arccos\braket{\omega}{\bar{\chi}_1}} = \frac{\pi}{2 \arccos\sqrt{\frac{n-3}{n-1}}} \\
    &= \frac{\pi}{2 \sqrt{2}} \sqrt{n} - \mathcal{O}(n^{-1/2})\, .
\end{align}
We have $\Lambda_1 = \{ n \}$, and it follows that $t_2 = \frac{\pi}{\gcd \Lambda_1} = \frac{\pi}{n}$. Given the second walk time, the next (and final) unitary is
\begin{equation}
    U_2 = \uw{\frac{\pi}{n}} (U_1)^{p_1} \, ,
\end{equation}
producing $\chi_2 = \emptyset$ and $\bar{\chi}_2 = \{ \ket{b_1}\}$. Consequently, 
\begin{align}
    p_2 &= \frac{\pi}{2\arccos\frac{\braket{\omega}{\bar{\chi}_2}}{\sqrt{ \braket{\omega}{s}^2 + \braket{\omega}{\bar{\chi}_2}^2}}} = \frac{\pi}{2 \arccos\sqrt{1 - \frac{1}{n}}} \\
    &= \frac{\pi}{2}\sqrt{n} - \mathcal{O}(n^{-1/2}) \, .
\end{align}
The overall state evolution for the $n=4m$ case uses a total of
\begin{equation}
    N_\text{iter} = \frac{1}{2}(p_1 p_2 - 1) = \frac{\pi ^2}{8 \sqrt{2}} n - \frac{1}{2} - \frac{\pi ^2}{8 \sqrt{2}} - \mathcal{O}(n^{-1})
\end{equation}
phase-walk iterations.
Comparing this to the complete graph, which for large $n$ uses $\frac{\pi}{4}\sqrt{\binom{n}{2}} = \frac{\pi}{4 \sqrt{2}} n$ iterations, we see that the alternating phase-walk framework requires approximately $\frac{\pi}{2}$ times more queries than the lower `Grover bound' on $J(4m,2)$ graphs. In fact, this is the worst-case $d=2$ scaling as per \cref{sec:efficiency}. We plot a numerical simulation of the search dynamics in \cref{fig:approximate-0mod4}. The oscillatory behaviour is indicative of the iteration of $U_1$ enveloped in the iteration of $U_2$.

\subsubsection{Case 3: $n \equiv 1 \mod 2$}

\begin{figure}
    \centering
    \includegraphics[width=\linewidth]{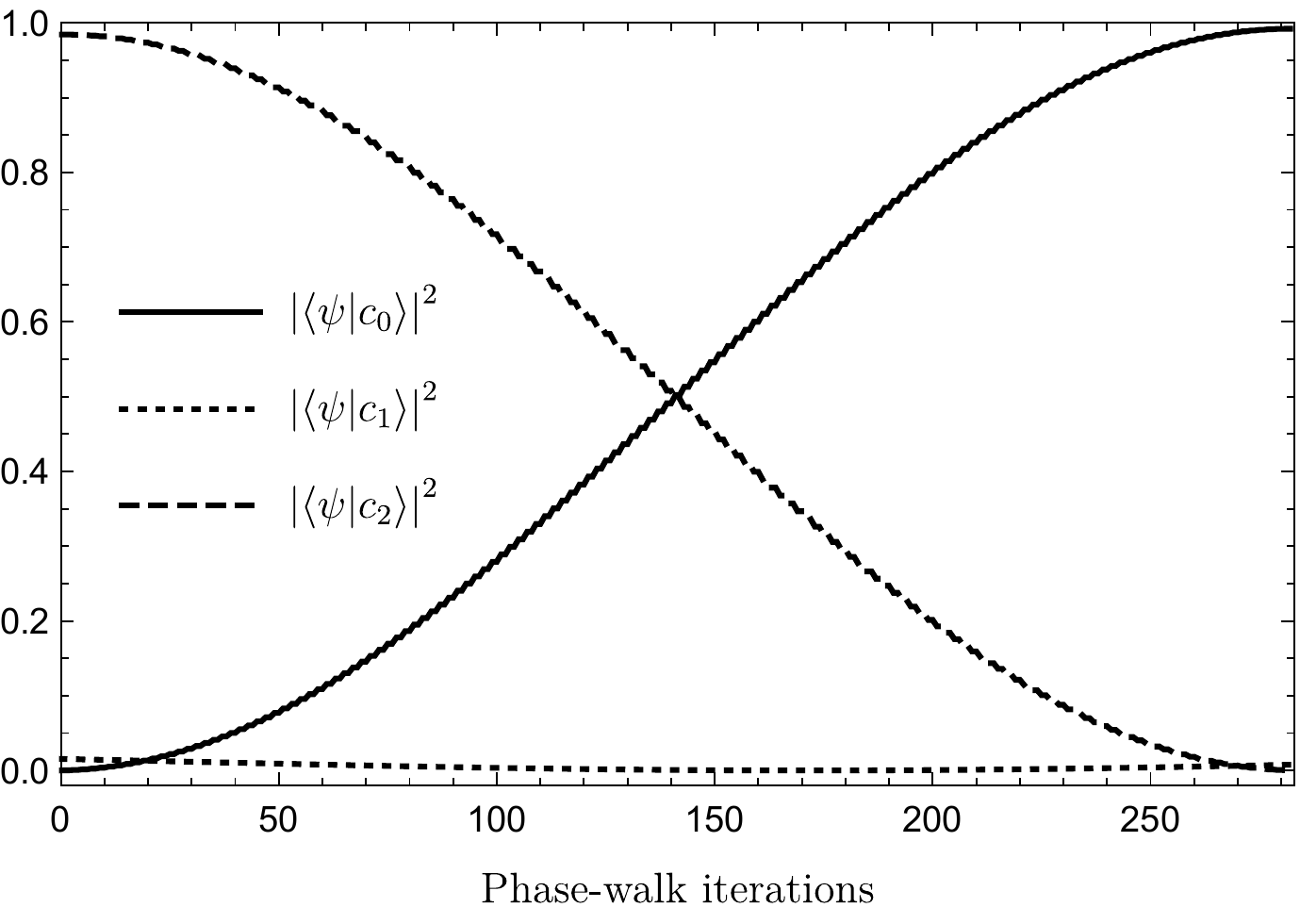}
    \caption{Numerical simulation of the search algorithm for a $J(257, 2)$ graph.}
    \label{fig:johnson-odd}
\end{figure}

% \begin{figure*}
%     \centering
%     \begin{subfigure}[t]{0.49\textwidth}
%         \centering
%         \includegraphics[width=\linewidth]{Johnson2D-257-Normal.pdf}
%         \caption{}
%     \end{subfigure}
%     \begin{subfigure}[t]{0.49\textwidth}
%         \centering
%         \includegraphics[width=\linewidth]{Johnson2D-257-Eig.pdf}
%         \caption{}
%     \end{subfigure}
%     \caption{Numerical simulation of the state evolution for a $J(257, 2)$ graph. The overlap is shown in (a) the walk basis, where $\ket{c_0} = \ket{\omega}$, and (b) the Laplacian eigenstate basis, where $\ket{b_0} = \ket{s}$.}
%     \label{fig:johnson-odd}
% \end{figure*}

The final case is when $n$ is odd, resulting in another $d=2$ case. For the $J(2m+1, 2)$ class of Johnson graphs, we have $\chi_1 = \{ \ket{b_2} \}$, $\bar{\chi}_1 = \{ \ket{b_1} \}$, $\chi_2 = \emptyset$, $\bar{\chi}_2 = \{ \ket{b_2} \}$. Here, $t_1 = \pi$ and $t_2 = \frac{\pi}{2(n-1)}$. Consequently, 
\begin{align}
    p_1 &=  \frac{\pi}{2 \arccos\braket{\omega}{\bar{\chi}_1}} = \frac{\pi}{2\arccos{\sqrt{\frac{2}{n}}}} = 1 + \mathcal{O}(n^{-1/2}) \, , \\
    p_2 &=\frac{\pi}{2\arccos\frac{\braket{\omega}{\bar{\chi}_2}}{\sqrt{ \braket{\omega}{s}^2 + \braket{\omega}{\bar{\chi}_2}^2}}} = \frac{\pi}{2\arccos{\sqrt{\frac{n(n-3)}{(n-1) (n-2)}}}} \\
    &= \frac{\pi}{2 \sqrt{2}}n -\frac{3 \pi }{4 \sqrt{2}} - \mathcal{O}(n^{-1}) \, .
\end{align}
The total iteration count using \cref{eq:niter} is therefore approximately
\begin{equation}
    N_\text{iter} = \frac{\pi}{4 \sqrt{2}}n+\frac{1}{2}\sqrt{n} - \mathcal{O}(1) \, .
\end{equation}
However, since $p_1 < 2$, we use the second approach to exact implementation of $U_1^{p_1}$ described in \cref{sec:exactu1}, making the total number of phase-walk iterations approximately three times this value.

As with the previous cases, we plot a numerical simulation of the search algorithm in \cref{fig:johnson-odd}. To make the plot more readable, only every second phase-walk iteration is shown. In reality, we have rapid oscillation every iteration due to $U_1^{p_1}$.

Thus, we see that $J(n,2)$ graphs can have either $d=1$ or $d=2$, depending on the value of $n \bmod 4$. In all cases, asymptotically optimal quantum search is achieved. This matches the behaviour of the \cg algorithm, which also has asymptotically optimal scaling on $J(n, 2)$ graphs \cite{Janmark2014,Wong2016}.

% Thus, we see that asymptotically optimal quantum spatial search is achievable using the alternating phase-walk formalism on $k=1$ and $k=2$ Johnson graphs.

\subsection{Rook graphs}

\begin{figure*}
    \centering
    \begin{subfigure}[t]{0.32\textwidth}
        \centering
        \includegraphics[width=\textwidth]{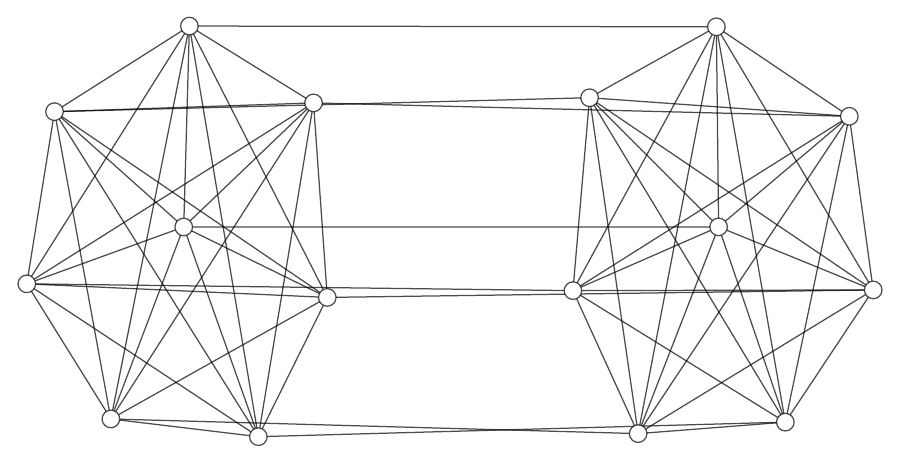}
        \caption{}
    \end{subfigure}
    \begin{subfigure}[t]{0.32\textwidth}
        \centering
        \includegraphics[width=\textwidth]{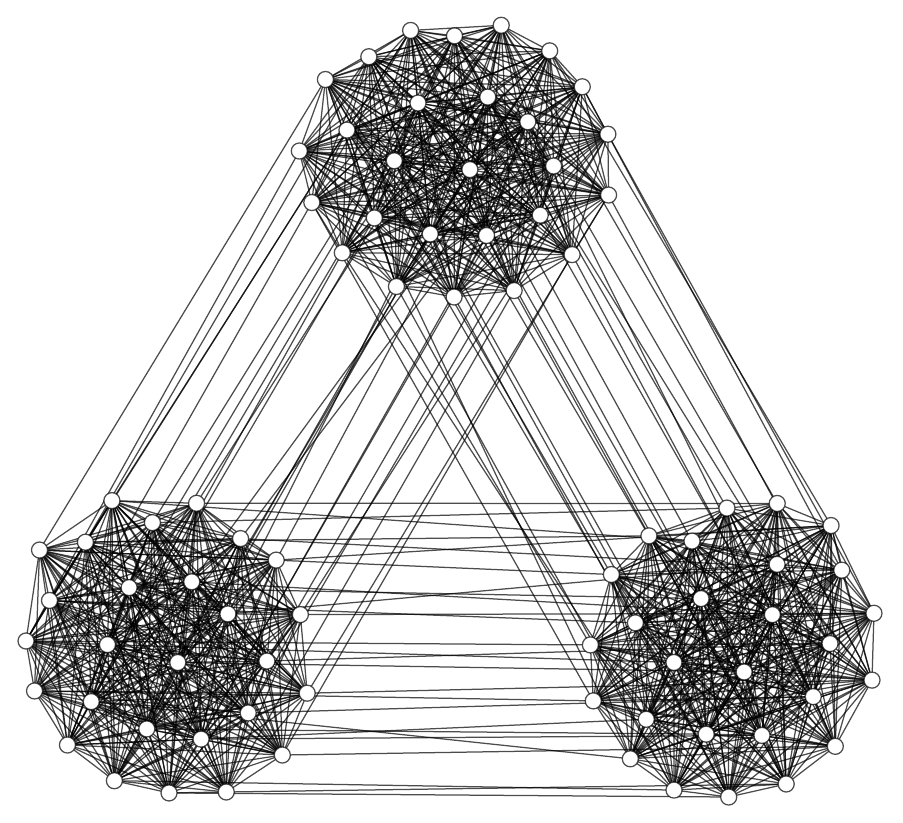}
        \caption{}
    \end{subfigure}
    \begin{subfigure}[t]{0.32\textwidth}
        \centering
        \includegraphics[width=\textwidth]{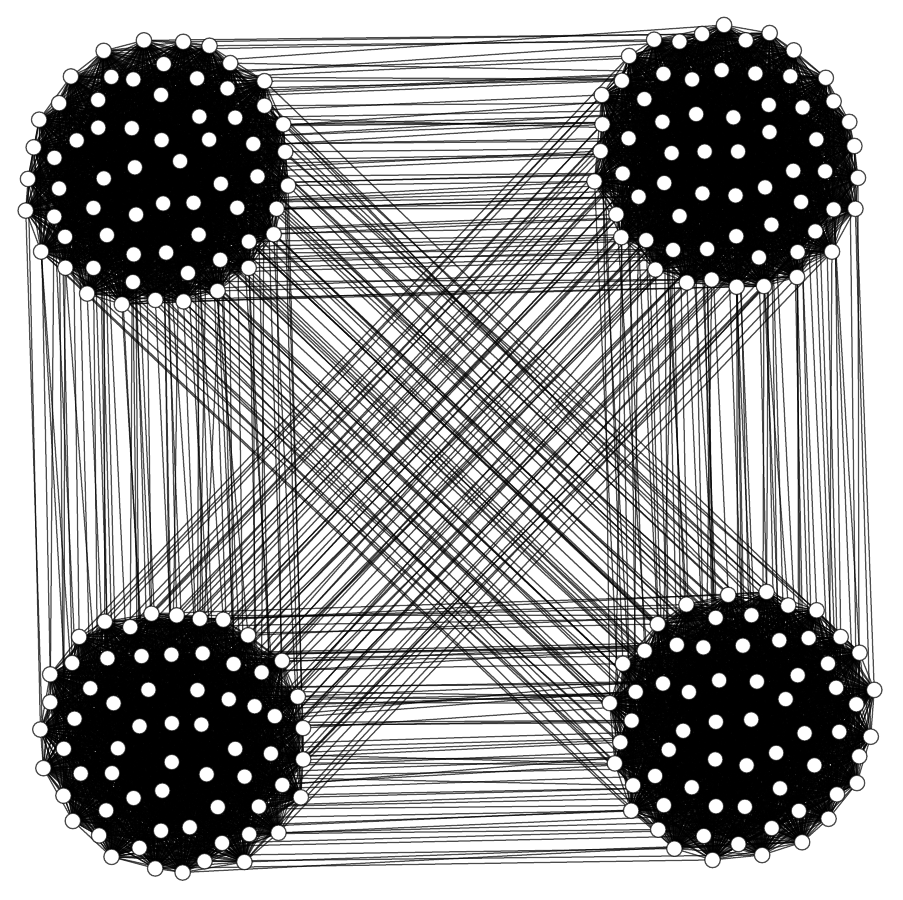}
        \caption{}
    \end{subfigure}
    \caption{The $n\times n^3$ rook graphs, which have a total of $N=n^4$ vertices, for (a) $n=2$ (b) $n=3$ (c) $n=4$. The $\mathcal{CG}$ algorithm scales suboptimally on this class of graphs, whilst the alternating phase-walk algorithm achieves $\mathcal{O}(1)$ marked-vertex overlap after approximately $\frac{\pi^2}{8}\sqrt{N}$ iterations.}
    \label{fig:rookgraph}
\end{figure*}

The $n_1 \times n_2$ rook graph is a particularly interesting case to study, due to the suboptimality of the $\mathcal{CG}$ algorithm for some ratios of $n_1$ to $n_2$~\cite{Chakraborty2020}. A $n_1 \times n_2$ rook graph is defined as the Cartesian product of two complete graphs $\mathbb{K}_{n_1} \cart \mathbb{K}_{n_2}$, having a total of $N=n_1 n_2$ vertices. The edges can be interpreted as the valid moves of a rook piece on a rectangular chessboard. The worst case for the \cg algorithm is when $n_1=n$ and $n_2 = n^3$, with this class of graphs shown in \cref{fig:rookgraph} for $n=2, 3, 4$. On these rectangular chessboards, the $\mathcal{CG}$ algorithm takes time $T=\mathcal{O}(N^{5/8})$ to produce a state having $\mathcal{O}(N^{-1/8})$ overlap with the marked vertex \cite{Chakraborty2020}. That is, the time taken is suboptimal, and the final state has diminishing overlap with the marked vertex. Here we show that by using the alternating phase-walk formulation, one can achieve $\mathcal{O}(1)$ overlap with the marked element after approximately $\frac{\pi ^2}{8} \sqrt{N}$ phase-walk iterations.

The Cartesian product of two vertex-transitive graphs is itself vertex-transitive, so we can analyse the search algorithm assuming $\ket{\omega} = \ket{0}$ without loss of generality. A simple reduced basis in which to study rook graphs is simply the tensor product of the two underlying reduced complete graph bases $\{ \ket{0}, \frac{1}{\sqrt{n_1 - 1}}\sum\limits_{j=1}^{n_1 - 1} \ket{j} \} \otimes \{ \ket{0}, \frac{1}{\sqrt{n_2 - 1}} \sum\limits_{j=1}^{n_2 - 1} \ket{j} \}$, producing
\begin{align}
    \ket{c_0} &= \ket{0} \otimes \ket{0} = \ket{\omega} \, ,\\
    \ket{c_1} &= \frac{1}{\sqrt{n_1 - 1}} \sum\limits_{j=1}^{n_1 - 1} \ket{0} \otimes \ket{j} \, ,\\
    \ket{c_2} &= \frac{1}{\sqrt{n_2 - 1}} \sum\limits_{j=1}^{n_2 - 1} \ket{j} \otimes \ket{0} \, , \\
    \ket{c_3} &= \frac{1}{\sqrt{(n_1-1)(n_2-1)}} \sum\limits_{j=1}^{n_1 - 1}\sum\limits_{k=1}^{n_2 - 1} \ket{j} \otimes \ket{k} \, .
\end{align}
Thus, the Laplacian can be expressed in this basis as
\begin{align}
    \mathcal{L} &= \begin{pmatrix}
    n_1 - 1 & -\sqrt{n_1 - 1} \\
    -\sqrt{n_1 - 1} & 1
    \end{pmatrix} \oplus \begin{pmatrix}
    n_2 - 1 & -\sqrt{n_2 - 1} \\
    -\sqrt{n_2 - 1} & 1
    \end{pmatrix}
\end{align}
where $A \oplus B = A \otimes \mathbb{I} + \mathbb{I} \otimes B$. The non-zero eigenvalues are $\Lambda_0 = \{ n_1, n_2, n_1 + n_2 \}$. Recall that for integral graphs, $d$ is the number of unique exponents of 2 in the prime factorisations of the non-zero eigenvalues. We can write $\Lambda_0 = \{2^{a} m_1, 2^b m_2 , 2^a ( m_1 + 2^{b-a} m_2) \}$ for odd $m_1$ and $m_2$, assuming without loss of generality that $b \geq a$. Thus $d=2$ for any pair of values $(n_1, n_2)$ and consequently the alternating phase-walk algorithm is asymptotically optimal on any rook graph based on the results from \cref{sec:success}, achieving $\mathcal{O}(1)$ overlap with the marked element after $\mathcal{O}(\sqrt{N})$ iterations.

For completeness, we continue with the analysis to explore the dynamics of the search algorithm on rook graphs. The overlaps of the marked vertex with the eigenstates of $\mathcal{L}$ are
\begin{align}
    \braket{\omega}{b_0} &= \braket{\omega}{s} = \frac{1}{\sqrt{n_1 n_2}} \, , \\
    \braket{\omega}{b_1} &= \sqrt{\frac{n_1 - 1}{n_1 n_2}} \, ,\\
    \braket{\omega}{b_2} &= \sqrt{\frac{n_2 - 1}{n_1 n_2}} \, ,\\
    \braket{\omega}{b_3} &= \sqrt{\frac{(n_1-1)(n_2-1)}{n_1 n_2}} \, .
\end{align}
When $n_1 = n$ and $n_2 = n^3$, $\gcd \Lambda = n$ so the first walk time is $t_1 = \frac{\pi}{n}$. In the eigenstate basis,
\begin{align}
    \uw{\frac{\pi}{n}} &= \exp(-i \pi \diag\left(0, 1, n^2, 1 + n^2\right)) \\
    &= \begin{cases}
        \diag\left( 1, -1, 1, -1\right) & n \equiv 0 \mod 2 \, , \\
        \diag\left( 1, -1, -1, 1\right) & n \equiv 1 \mod 2 \, .
    \end{cases}
\end{align}
We consider the even $n$ case, where $\chi_1 = \{\ket{b_2}\}$ and $\bar{\chi}_1 = \{ \ket{b_1}, \ket{b_3} \}$. We have 
\begin{align}
    \braket{\omega}{\chi_1} &= \braket{\omega}{b_2} = \sqrt{\frac{n^3 - 1}{n^4}} \, , \\
    \braket{\omega}{\bar{\chi}_1} &= \sqrt{\braket{\omega}{b_1}^2 + \braket{\omega}{b_3}^2} = \sqrt{1 - \frac{1}{n}} \, .
\end{align}
It also follows that $t_2 = \frac{\pi}{n^3}$, with $\chi_2 = \emptyset$ and $\bar{\chi}_2 = \{ \ket{b_2 }\}$. Hence,
\begin{align}
    p_1 &= \frac{\pi}{2\arccos{\sqrt{1 - \frac{1}{n}}}} = \frac{\pi}{2}\sqrt{n} - \mathcal{O}(n^{-1/2}) \, , \\
    p_2 &= \frac{\pi}{2\arccos{\frac{\braket{\omega}{\bar{\chi_2}}}{\sqrt{\braket{\omega}{s}^2 + \braket{\omega}{\chi_1}^2}}}} = \frac{\pi}{2\arccos{\sqrt{1 - \frac{1}{n^3}}}} \\
    &= \frac{\pi}{2}n^{3/2} - \mathcal{O}(n^{-3/2}) \, ,
\end{align}
and the total number of iterations is approximately
\begin{equation}
    N_\text{iter} = \frac{1}{2}(p_1 p_2 - 1) = \frac{\pi^2}{8}n^2 - \mathcal{O}(n) = \frac{\pi^2}{8}\sqrt{N} - \mathcal{O}(N^{1/4}) \, .
\end{equation}
% Thus, as a $d=2$ periodic graph, this is a case where the alternating phase-walk framework achieves $\mathcal{O}(\sqrt{N})$ search with $\mathcal{O}(1)$ overlap.
We plot the measurement probabilities against the number of iterations for an $8 \times 8^3$ rook graph in \cref{fig:rook-plot}, where again the oscillatory pattern is indicative of the iteration of $U_1$ enveloped in the iteration of $U_2$.

% Thus, in this section we have shown that the alternating phase-walk formalism can achieve asymptotically optimal search on $n \times n^3$ Rook graphs where the \cg algorithm cannot.

% Observe also the spike in probability in the final iterations, where $(U_1)^{\nint{(p_1-1)/2}}$ is applied.

\begin{figure}
    \centering
    \includegraphics[width=\linewidth]{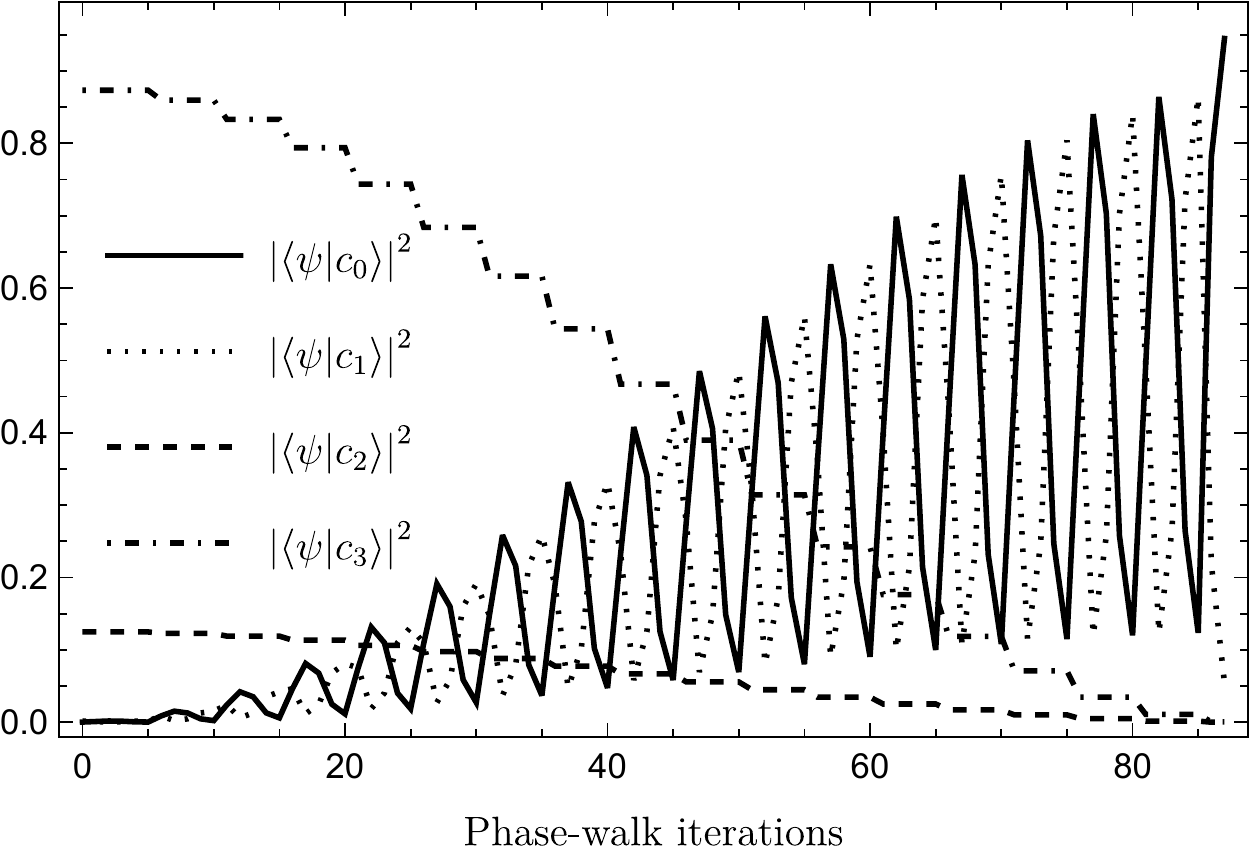}
    \caption{The state evolution of the alternating phase-walk spatial search algorithm on an $8 \times 8^3$ rook graph.}
    \label{fig:rook-plot}
\end{figure}

\subsection{Complete-square graph}

\begin{figure*}
    \centering
    \begin{subfigure}[t]{0.35\textwidth}
        \centering
        \includegraphics[width=\linewidth]{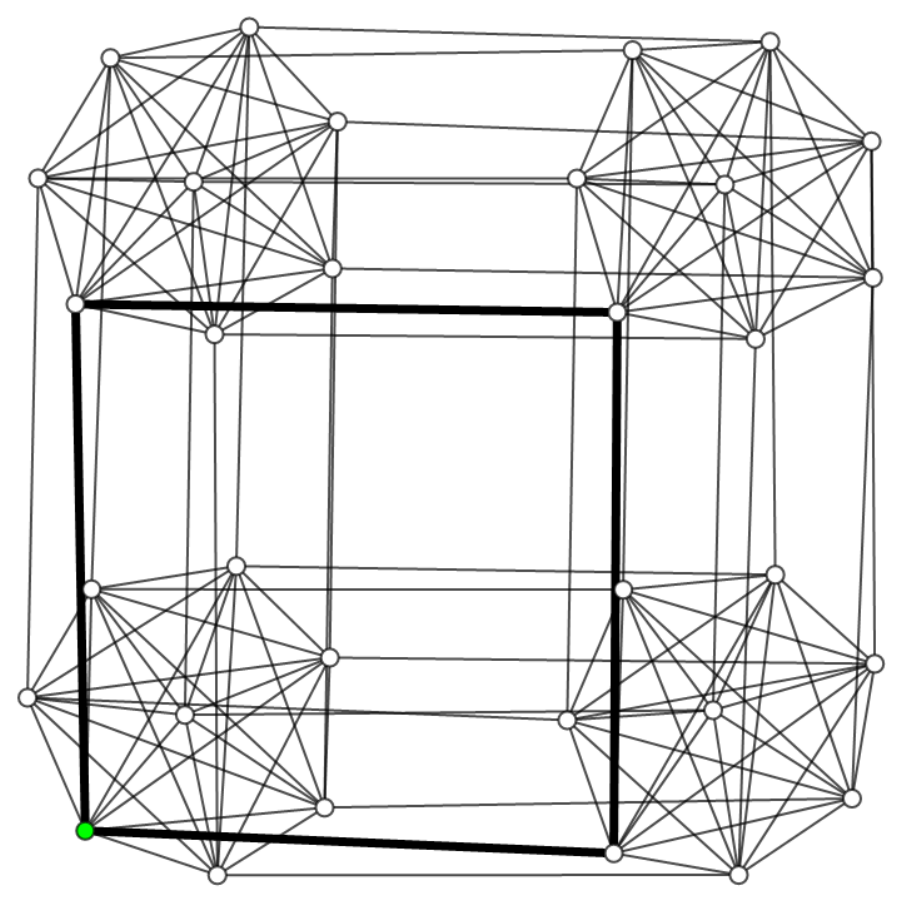}
        \caption{}
    \end{subfigure}
    \qquad
    \begin{subfigure}[t]{0.35\textwidth}
        \centering
        \includegraphics[width=\linewidth]{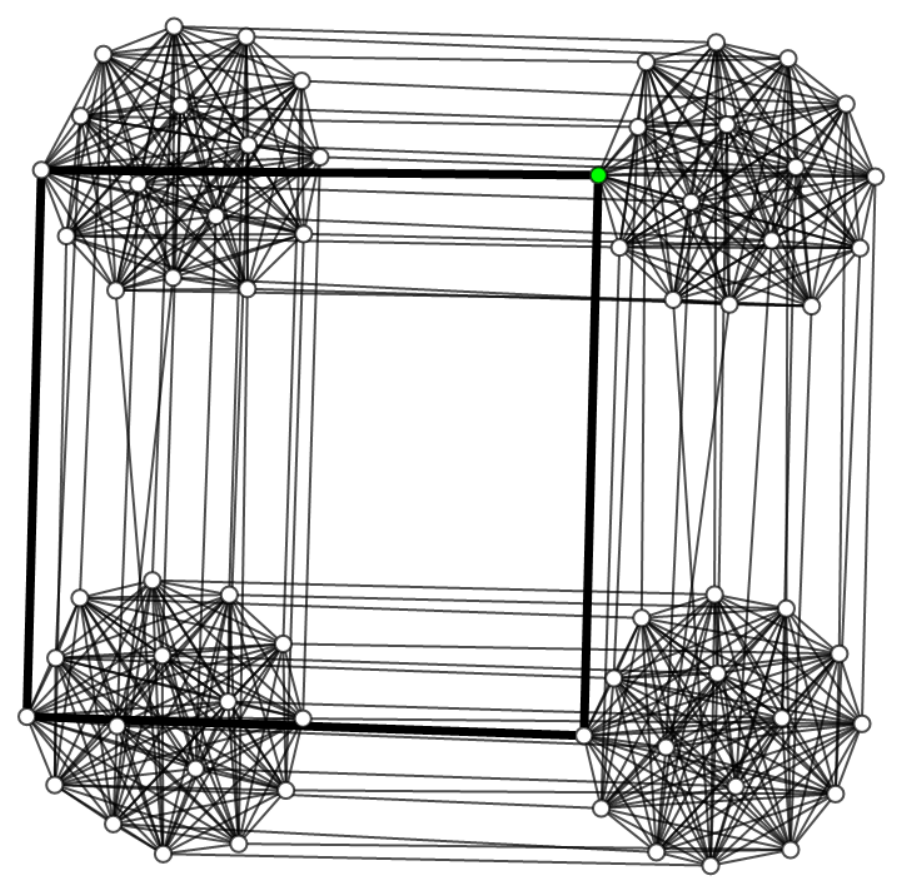}
        \caption{}
    \end{subfigure}
    \caption{Two examples of a  $\mathbb{K}_n \cart \mathbb{Q}_2$ graph, for (a) $n=8$ and (b) $n=16$. The bolded edges indicate the connectivity of a marked vertex (green) to its associated square subgraph.}
    \label{fig:completecube}
\end{figure*}

As a final example, we consider a $d=3$ case that illustrates the intriguing relationship of the alternating phase-walk framework to perfect state transfer, which has a strong connection to periodicity \cite{Christandl2005PerfectTO,Godsil2012StateTO,Kay2010}. Consider the Cartesian product of an $n$-vertex complete graph $K_n$ with a square graph $\mathbb{Q}_2$, denoted $\mathbb{K}_n \cart \mathbb{Q}_2$. An illustration of the $\mathbb{K}_n \cart \mathbb{Q}_2$ graphs is shown in \cref{fig:completecube}. As another vertex-transitive graph, we can take $\ket{0}$ to be the marked element without loss of generality. Performing standard dimensionality reduction, the resultant search basis is simply the tensor product of the two underlying reduced bases
\begin{equation}
    \{ \ket{0}, \frac{1}{\sqrt{n-1}} \sum\limits_{x=1}^{n-1} \ket{x} \} \otimes \{ \ket{0}, \frac{1}{\sqrt{2}} \left( \ket{1} + \ket{2} \right), \ket{3} \} \, .
\end{equation}
In this basis, the reduced Laplacian of $\mathbb{K}_n \cart \mathbb{Q}_2$ is the Kronecker sum of the reduced complete graph Laplacian and the reduced square Laplacian,
% giving
% \begin{align}
%     \ket{c_0} &= \ket{0}  = \ket{\omega} \\
%     \ket{c_1} &= \frac{1}{\sqrt{2}}\left( \ket{1} + \ket{2} \right) \\
%     \ket{c_2} &= \ket{3} \\
%     \ket{c_3} &= \frac{1}{\sqrt{n-1}} \sum\limits_{x=1}^{n-1} \ket{4x} \\
%     \ket{c_4} &= \frac{1}{\sqrt{n-1}} \sum\limits_{x=1}^{n-1} \left(\ket{4x+1} + \ket{4x+2}\right) \\
%     \ket{c_5} &= \frac{1}{\sqrt{n-1}} \sum\limits_{x=1}^{n-1} \ket{4x+3}
% \end{align}
\begin{equation}
    \mathcal{L} = \begin{pmatrix}
    n-1 & -\sqrt{n-1} \\
    -\sqrt{n-1} & 1
    \end{pmatrix}
    \oplus
    \begin{pmatrix}
    2 & -\sqrt{2} & 0 \\
    -\sqrt{2} & 2 & -\sqrt{2} \\
    0 & -\sqrt{2} & 2
    \end{pmatrix} \, .
\end{equation}
Calculating the non-zero eigenvalues gives $\Lambda_0 = \{2,4,n,n+2,n+4\}$, and the marked vertex Laplacian overlaps are
\begin{align}
    \braket{\omega}{b_0} &= \frac{1}{\sqrt{2}}\braket{\omega}{b_1} = \braket{\omega}{b_2} = \frac{1}{\sqrt{N}} = \frac{1}{2\sqrt{n}} \, ,\\
    \braket{\omega}{b_3} &= \frac{1}{\sqrt{2}} \braket{\omega}{b_4}= \braket{\omega}{b_5} = \frac{1}{2} \sqrt{1 - \frac{1}{n}} \, .
\end{align}
For brevity the definitions of the Laplacian eigenstates $\ket{b_k}$ are omitted, since they are easily computed from $\mathcal{L}$.
We restrict to the case where $n \equiv 0 \mod 8$. Here $\gcd \Lambda_0 = 2$, so $t_1 = \frac{\pi}{2}$. Thus, in the reduced Laplacian eigenstate basis,
\begin{align}
    U_w(\frac{\pi}{2}) = \diag\left(1,-1,1,1,-1,1\right) \, .
\end{align}
It immediately follows that $\chi_1  = \{ \ket{b_2}, \ket{b_3}, \ket{b_5}\}$ and $\bar{\chi}_1 = \{ \ket{b_1}, \ket{b_4}\}$ with
\begin{align}
    \braket{\omega}{\chi_1} &= \sqrt{\braket{\omega}{b_2}^2 + \braket{\omega}{b_3}^2 + \braket{\omega}{b_5}^2} = \frac{1}{2} \sqrt{2-\frac{1}{n}} \, , \\
    \braket{\omega}{\bar{\chi}_1} &=  \sqrt{\braket{\omega}{b_1}^2 + \braket{\omega}{b_4}^2} = \frac{1}{\sqrt{2}} \, .
\end{align}
Thus
\begin{align}
    p_1 &= \frac{\pi}{2\arccos{\braket{\omega}{\bar{\chi}_1}}}
    = 2 \, ,
\end{align}
and we find the first iteration count is always $2$ independent of the graph size.
Now with $\Lambda_1 = \{4,n,n+4\}$, $\gcd \Lambda_1 = 4$ and so $t_2 = \frac{\pi}{4}$. For the $n \equiv 0 \mod 8$ case, we have $\chi_2 = \{ \ket{b_3}\}$ and $\bar{\chi}_2 = \{ \ket{b_2}, \ket{b_5}\}$. Hence,
\begin{align}
    \braket{\omega}{\chi_2} &= \braket{\omega}{b_3} = \frac{1}{2} \sqrt{1 - \frac{1}{n}} \,, \\
    \braket{\omega}{\bar{\chi}_2} &= \sqrt{\braket{\omega}{b_2}^2 + \braket{\omega}{b_5}^2} = \frac{1}{2} \, .
\end{align}
Then
\begin{equation}
    p_2 = \frac{\pi}{2\arccos{\frac{\braket{\omega}{\bar{\chi}_2}}{\sqrt{\braket{\omega}{s}^2 + \braket{\omega}{\chi_1}^2}}}} = 2 \, ,
\end{equation}
and again we find that the second iteration count is always $2$, independent of the number of vertices.

Finally, $\lambda_2 = \{ n \} \implies t_3 = \frac{\pi}{n}$ and $\chi_3 = \emptyset$, $\bar{\chi}_3 = \{ \ket{b_3}\}$. Hence,
\begin{align}
    p_3 &= \frac{\pi}{2\arccos{\frac{\braket{\omega}{\bar{\chi}_3}}{\sqrt{\braket{\omega}{s}^2 + \braket{\omega}{\chi_2}^2}}}} = \frac{\pi }{2 \arccos\sqrt{1 - \frac{1}{n}}} \\
    &= \frac{\pi}{2}\sqrt{n} - \mathcal{O}(n^{-1/2}) \, .
\end{align}
Using \cref{thm:genproof}, we have the state evolution $U_1^{\nint{(p_1-1)/2}}U_2^{\nint{(p_2-1)/2}}U_3^{\nint{(p_3-1)/2}}\ket{s}$, but $\frac{1}{2}(p_1 - 1) = \frac{1}{2}(p_2 - 1) = \frac{1}{2}$, so we round down to 0 iterations of $U_1$ and $U_2$. This draws parallels to the $d=2$ worst-case error scenario described in \cref{sec:efficiency}. A closer look at the resulting evolution gives
\begin{align}
    U_3^{(p_3-1)/2}\ket{s} &= \ket{\omega_2} = \frac{1}{\sqrt{n}}\ket{s} + \sqrt{1 - \frac{1}{n}} \ket{b_3} \\
    &= \ket{0} \otimes \frac{1}{2} \left(\ket{0} + \ket{1} + \ket{2} + \ket{3}\right)\,
\end{align}
which is precisely the square subgraph containing the marked vertex. Thus, measuring this state has 25\% probability of producing the marked vertex, and more generally gives a strong classical hint as to the location of the marked vertex -- it must be one of the 4 vertices on the associated square. Alternatively, we can utilise the $U_f(\theta)$ parametrisation technique from \cref{sec:exactu1} to find that
\begin{equation}
    \ket{\omega} = U_w(t_2) U_w(t_1) U_f(\frac{\pi}{2}) U_w(t_1) U_f(-\frac{\pi}{2}) \ket{\omega_2} \, .
\end{equation}
Thus, with two additional controlled phase rotations we establish an $1 - \mathcal{O}(\frac{1}{N})$ overall success probability, requiring
\begin{equation}
    N_\text{iter} = \frac{1}{2} p_1 p_2 (p_3 - 1) + 2 = 2 p_3  = \frac{\pi}{2}\sqrt{N} - \mathcal{O}(N^{-1/2})
\end{equation}
phase-walk iterations.
As with the previous examples, the state after each iteration of $U_3$ is plotted in \cref{fig:evolvecompletesquare} using the reduced basis. Here the amplification of the `marked square' is clearly visible.

Moreover, the search procedure has an interesting interpretation in terms of perfect state transfer. Observe that $\uw{\frac{\pi}{2}}\ket{0} = \ket{3}$, i.e. there is perfect state transfer between a vertex and the `opposite' vertex on its square subgraph. Perfect state transfer on $\mathbb{K}_n \cart \mathbb{Q}_2$ is inherited from the same property on the $m$-dimensional hypercube $\mathbb{Q}_m$. With this in mind, the first three steps in the iterate $U_1^{p_1}=\uw{\frac{\pi}{2}}\uf{\pi} \uw{\frac{\pi}{2}}\uf{\pi}$ can be seen to (1) flip the phase of the marked vertex, (2) transfer this phase to the opposite vertex on the square, and (3) re-introduce a phase flip to the marked vertex. Thus, we essentially have a unitary that marks two vertices on the graph using perfect state transfer. Similarly, $U_2$ propagates the phase flip to the other two corners of the square. Hence, each iteration of $U_3$ amplifies not only the marked vertex, but the associated square subgraph. Thus we see that this interesting property can naturally emerge from the alternating phase-walk framework, where perfect state transfer is utilised to amplify a subset of vertices that includes the marked element. 

\begin{figure}
    \centering
    \includegraphics[width=\linewidth]{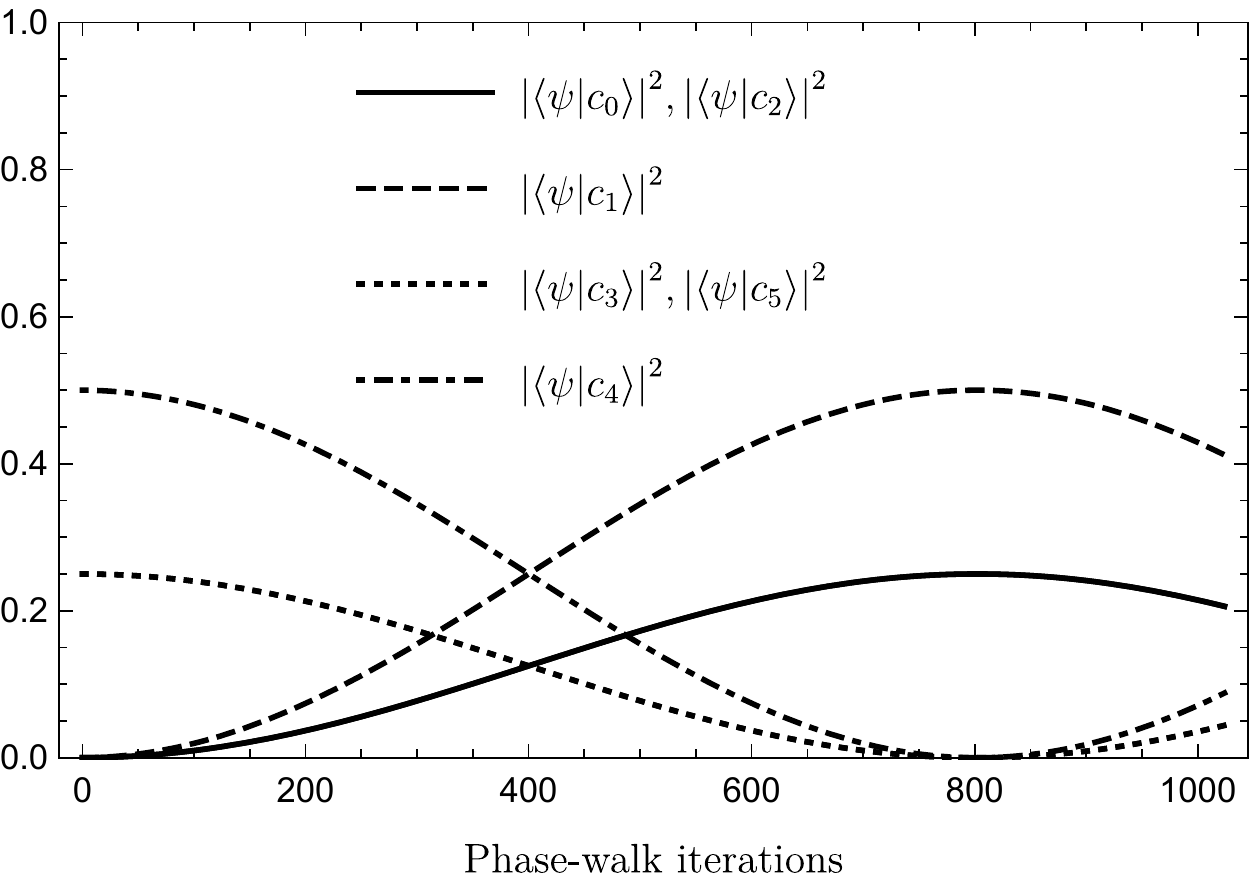}
    \caption{The state evolution on a $\mathbb{K}_n \cart \mathbb{Q}_2$ graph with 262,144 vertices, where $\ket{c_0}$ is the marked vertex, $\ket{c_2}$ is its `opposite' vertex, and $\ket{c_2}$ is the equal superposition of the two other vertices on the associated square subgraph.}
    \label{fig:evolvecompletesquare}
\end{figure}

\section{Discussion and Conclusion}

In this paper, we have proposed a new methodological framework for quantum spatial search on periodic graphs, where quantum walks and marked vertex phase shifts are interleaved in an alternating fashion. We prove a unitary evolution that maps the initial superposition to the marked vertex, where the necessary walk times can be calculated in closed form. We show that the number of iterations is $\mathcal{O}(\sqrt{N})$ if a graph parameter $d$, which is bounded above by the number of unique eigenvalues, is independent of $N$. For the $d=1$ case, the evolution is equivalent to that of Grover's algorithm. We show that the evolution can be made partially error-free without sacrificing efficiency, hinting at the possibility of a fully deterministic and asymptotically optimal spatial search algorithm for any $d$. Utilising this modification, we show the success probability in the $d=2$ case is always $\mathcal{O}(1)$. Finally, we demonstrate the algorithm's performance on a number of different graph classes. Of particular note is the $n \times n^3$ rectangular rook graph, where it is known that the \cg algorithm is suboptimal. We demonstrate that under the alternating phase-walk formalism, rook graphs meet the $d=2$ criteria and thus a marked vertex is located with $\mathcal{O}(1)$ probability using $\mathcal{O}(\sqrt{N})$ iterations. We also show by example that the algorithm can naturally utilise perfect state transfer, in order to amplify a subgraph containing the marked vertex.

Our work motivates additional study into the alternating phase-walk formalism, in terms of further characterising the set of graphs where efficient spatial search is achievable, and drawing connections to the \cg algorithm. The restriction to periodic graphs is natural, and provokes study into whether periodic graphs are the \textit{only} graphs that admit efficient quantum search using alternating phase-walks.

The alternating phase-walk framework also has the advantage of being amenable to quantum circuit implementation, in addition to `fundamental' quantum spatial search on a structure as per the \cg algorithm. For some graphs, the time-evolution of a quantum walk can be fast-forwarded \cite{Loke2017,Atia2017}, leading to an efficient quantum circuit for database search. An interesting topic for future research is characterising which periodic graphs admit fast-forwarded quantum simulation for the purposes of circuit-based alternating phase-walk search.

Another compelling direction for future research is making the algorithm deterministic. Our previous work and numerical evidence hints that with careful choice of marked vertex phase shifts, it will be possible to achieve deterministic search using the alternating phase-walk approach for any periodic graph. Furthermore, we also point out that a deterministic spatial search algorithm $\mathcal{A}_\omega$ implies a kind of perfect state transfer between two vertices $\ket{\omega_1}$ and $\ket{\omega_2}$ by computing \begin{equation}
    e^{i \phi} \ket{\omega_2} = \mathcal{A}_{\omega_2} \mathcal{A}_{\omega_1}^\dag \ket{\omega_1} \, ,
\end{equation}
making this a promising idea for future study.

Finally, we briefly comment on the implications of this work to quantum combinatorial optimisation and relevant future directions. It is likely that the performance of a quantum mixer on an arbitrary optimisation problem is bounded by its performance on the database search problem. Therefore an appropriate choice of mixer is one which obtains the optimal $\mathcal{O}(\sqrt{N})$ scaling on the search problem, based on the spectral criteria described in this paper. Another significant pathway for future work is to consider the generalised phase unitary $U_f(\theta)=e^{-i \theta \sum_{x=0}^{N-1} f(x) \ket{x} \bra{x}}$ and replace the search objective function $f \colon \{ 0, 1, \ldots N-1 \} \mapsto \{ 0, 1\}$ with a generic objective function $f \colon \{ 0, 1, \ldots N-1 \} \mapsto \{ 0, 1, \ldots M-1\}$, $M > 0$. By performing a similar analysis to that of this paper, there is the potential to determine optimal values for (or relationships between) the evolution parameters $\Vec{t}$ and $\Vec{\theta}$ in this more general case. This would reduce or eliminate the need for classical variational optimisation to determine parameter values, which suffers from the curse of dimensionality as the number of parameters increases \cite{PhysRevA.97.022304}.

\acknowledgements

This research was supported by a Hackett Postgraduate Research Scholarship and an Australian Government Research Training Program Scholarship at The University of Western Australia. We thank Leonardo Novo for valuable insight and suggestions.  We would also like to thank Lyle Noakes, Michael Giudici, Gorden Royle, John Bamberg, Caiheng Li and Cheryl Praeger for mathematical and graph-theoretical discussions.  

\bibliography{refs.bib}

\end{document}